\documentclass{lmcs}
\pdfoutput=1

\usepackage{lastpage}
\lmcsdoi{19}{3}{7}
\lmcsheading{}{\pageref{LastPage}}{}{}%
{Feb.~15,~2022}{Jul.~25,~2023}{}

\usepackage[hidelinks,breaklinks]{hyperref}
\usepackage{cleveref}

\crefname{defi}{Definition}{Definitions} 
\crefname{rem}{Remark}{Remarks}
\crefname{exa}{Example}{Examples}

\crefname{thm}{Theorem}{Theorems}
\crefname{cor}{Corollary}{Corollaries}
\crefname{lem}{Lemma}{Lemmas}

\crefname{thmC}{Theorem}{Theorems}
\crefname{lemC}{Lemma}{Lemmas}

\usepackage[utf8]{inputenc}
\usepackage{amsthm}
\usepackage{amsmath}
\usepackage{amsfonts}
\usepackage{amssymb}
\usepackage{stmaryrd}
\usepackage{xcolor}
\usepackage{bm}

\usepackage{pgf,tikz}
\usepackage{amsthm}

\usepackage{xspace}
\usepackage{float}
\usetikzlibrary{positioning,automata,shapes,arrows}

\usepackage{ifthen}
\newboolean{shortver}
\setboolean{shortver}{false}

\newcommand{\shortlong}[2]{\ifthenelse{\boolean{shortver}}{#1}{#2}}

\shortlong{\input{applabels}}{}

\newcommand{\N}{\mathbb N}

\newcommand{\trans}[1]{\stackrel{#1}{\rightarrow}}
\definecolor{evegreen}{rgb}{0.3, .95, 0.4}
\definecolor{adamred}{rgb}{1, 0.4, 0.3}

\newcommand{\A}{\mathcal A}
\newcommand{\B}{\mathcal B}
\newcommand{\cl}[1]{#1^\uparrow}
\newcommand{\Ast}{A^*}
\newcommand{\first}{\mathit{first}}
\newcommand{\last}{\mathit{last}}

\newcommand{\dom}{\mathit{dom}}

\newcommand{\FV}{\mathrm{FV}}
\newcommand{\FOp}{\mathrm{FO}^+}
\newcommand{\sem}[1]{\llbracket#1\rrbracket}
\newcommand{\semphi}{\sem{\varphi}}
\newcommand{\sempsi}{\sem{\psi}}
\newcommand{\qr}{\mathrm{qr}}
\newcommand{\psitot}{\psi_{\mathit{tot}}}
\newcommand{\philoc}{\varphi_{\mathit{loc}}}
\newcommand{\phitrans}{\varphi_C}

\newcommand{\psim}{\psi_-}
\newcommand{\psip}{\psi_+}
\newcommand{\phiK}{\varphi_K}
\newcommand{\phiKG}{(\varphi_K)^G}
\newcommand{\Gw}{\mathcal{G}_w}
\newcommand{\GK}{\mathcal{G}_K}

\newcommand\Sq{\square}
\newcommand\Ci{\bigcirc}
\newcommand\ME{\mathit{M}}

\newcommand\C{\mathcal C}
\newcommand\ES{E_\circ}
\newcommand\diam{\lozenge}

\newcommand{\EF}{\mathrm{EF}^+}
\newcommand{\EFn}{\EF_n}
\newcommand{\Eve}[1]{\tikz\node [circle,draw=black, fill=evegreen!60, inner sep=0,minimum size=.34cm] {{\small #1}};}
\newcommand{\Adam}[1]{\tikz\node [draw,fill=adamred!60,inner sep=0, minimum width=.3cm,minimum height=.3cm] {{\small #1}};}

\newcommand\duo[2]{{{#1} \choose {#2}}}
\newcommand\x{\duo ab}
\newcommand\y{\duo bc}
\newcommand\z{\duo ca}
\newcommand\dl\duo

\newcommand{\Dir}{\{\leftarrow,\rightarrow\}}
\newcommand{\type}{\mathit{type}}
\newcommand{\Abase}{A_\mathit{base}}
\newcommand{\Lbase}{L_\mathit{base}}
\newcommand{\LM}{L_M}
\newcommand{\Aamb}{A_\mathit{amb}}
\newcommand{\add}{a_\delta^{\delta'}}
\newcommand{\Sigmabase}{\Sigma_\mathit{base}}
\newcommand{\Sigmaamb}{\Sigma_\mathit{amb}}

\begin{document}

\title{Positive First-order Logic on Words and Graphs\rsuper*}
\author[D.~Kuperberg]{Denis Kuperberg}	
\address{CNRS, LIP, ENS Lyon}	
\email{denis.kuperberg@ens-lyon.fr}  

\titlecomment{{\lsuper*}This paper is an extended version of \cite{Kup21}, the main new content is the section about graphs.}

\begin{abstract}
We study $\FOp$, a fragment of first-order logic on finite words, where monadic predicates can only appear positively. We show that there is an FO-definable language that is monotone in monadic predicates but not definable in $\FOp$. This provides a simple proof that Lyndon's preservation theorem fails on finite structures. We lift this example language to finite graphs, thereby providing a new result of independent interest for FO-definable graph classes: negation might be needed even when the class is closed under addition of edges. We finally show that the problem of whether a given regular language of finite words is definable in $\FOp$ is undecidable.
\end{abstract}

\maketitle

\section{Introduction}
Preservation theorems in first-order logic (FO) establish a link between semantic and syntactic properties \cite{AlGur97,Rossman08}.
We will be particularly interested here in Lyndon's theorem \cite{Lyndon59}, which states that if a first-order formula is monotone in a predicate $P$ (semantic property), then it is equivalent to a formula that is positive in $P$ (syntactic property). Recall that ``monotone in $P$'' means that the formula stays true when tuples are added to $P$, while ``positive in $P$'' means that $P$ does not appear under a negation in the formula.

As it is often the case with preservation theorems, Lyndon's Theorem may not hold when restricting the class of structures considered. Whether Lyndon's Theorem is true when restricted to finite structures was an open problem for 28 years. It was finally shown to fail on finite structures in \cite{AjtaiGurevich87} with a very difficult proof,  using a large array of techniques from different fields of mathematics such as probability theory, topology, lattice theory, and analytic number theory.
A simpler but still quite intricate proof of this fact was later given by \cite{Stol95}, using Ehrenfeucht-Fra\"iss\'e games on grid-like structures equipped with two binary predicates. This construction was slightly modified in \cite{RosenPHD} to treat a signature monotone in every relation symbol.
\medskip

The goal of this paper is to further restrict the class of structures under consideration, starting by allowing only finite words. This will allow us to obtain in turn a better understanding of the problem for finite graphs and general finite structures. We will therefore work in most of this paper with the particular signature associated with finite words: one binary predicate (the total order on positions in the word), and a finite set of monadic predicates (encoding the alphabet). We will call $\FOp$ the fragment of first-order logic on these models, that is syntactically positive in the monadic predicates. Or more simply: $\FOp$ is the negation-free fragment of FO on finite words.\medskip

Our purpose is twofold:
\begin{itemize}

\item Find out whether Lyndon's Theorem holds on finite words, and investigate the relation of this framework with the more general case of finite structures, as well as finite graphs. 
\item From the point of view of language theory: study the natural fragment of $\FOp$-definable languages, in particular given a regular language, can we decide whether it is $\FOp$-definable?
\end{itemize}


Recall that FO on words is a well-studied logic defining a proper fragment of regular languages. This fragment has many equivalent characterizations: definable by star-free expressions, aperiodic monoids, LTL formulas,... \cite{DG08,Schutz,Kamp,McNaughtonPapert}. In particular, given a regular language, it is decidable whether it is FO-definable.

\subsection*{Contributions}
We define a semantic notion of monotone language on alphabets equipped with a partial order: the language is required to be closed under replacement of a letter by a bigger one. This generalizes the monotonicity condition on monadic predicates in the sense of Lyndon. The negation-free logic $\FOp$ can only define monotone languages, and can be seen as a fragment of the standard FO logic on words, in the context of ordered alphabets.

Answering our first objective, we show that Lyndon's Theorem fails on finite words, by building a regular language that is monotone and FO-definable, but not $\FOp$-definable. This proof uses a variant of Ehrenfeucht-Fra\"iss\'e games that characterizes $\FOp$-definability, introduced in \cite{Stol95}, and instantiated here on finite words. As a corollary, using suitable axiomatizations of finite words, we obtain the failure of Lyndon's theorem on finite structures, in a much simpler way than in \cite{AjtaiGurevich87,Stol95}. We show that this can be done regardless of the precise version of Lyndon's preservation theorem that is considered: either monotonicity is with respect to one predicate, or an arbitrary subset of them, or all of them. These different versions are discussed in \cite{Lyndon59,AjtaiGurevich87,Stol95}. The version where all predicates are monotone, can be reformulated as closure under surjective homomorphisms \cite{Lyndon59}.

We also show that our counter-example language can be lifted to finite graphs, both directed and undirected, using a suitable encoding. We thereby show that there exists an FO-definable class of graphs, closed under edge addition, but such that any FO formula defining this class must use some edge predicate under a negation. To our knowledge, this is a new result, and we believe it is of independent interest, as it provides a surprising answer to a natural question about graphs.

Finally, answering our second objective, we show that $\FOp$-definability is undecidable for regular languages. This result is obtained using a reduction from the Turing Machine Mortality problem \cite{Hooper66}. To our knowledge, this is the first example of a natural\footnote{The concept of ``natural'' class is of course a bit informal here, but for instance we can think of it as classes inductively defined via a syntax. More generally, any class of regular languages not purposely defined to have an undecidable membership problem could be considered natural in this context.} class of regular languages for which membership is undecidable.

Although we work in a specialized framework requiring a partial order on the alphabet, we believe that this is not an artificial construct, as it occurs naturally in several settings. On one hand, powerset alphabets -- i.e. letters are sets of atomic predicates, and are naturally ordered by inclusion -- are standard in verification and model theory. It is of interest to remark that all the results of this paper are true in the particular case of powerset alphabets, and that this allows to obtain the failure of Lyndon's theorem on general finite structures.  On the other hand, ordered letters can be used as an abstraction for factors of the input word that are naturally equipped with a partial order, for instance in quantitative generalizations of regular languages such as regular cost functions \cite{CostFun}.

\subsection*{Related works}~\\
\noindent\textbf{Monotone complexity:}

Positive fragments of first-order logic play a prominent role in complexity theory. Indeed, an active research program consists in studying positive fragments of complexity classes. This includes for instance trying to lift equivalent characterizations of a class to their positive versions, or investigating whether a semantic and a syntactic definition of the positive variant of a class are equivalent. See \cite{GrigniSipser} for an introduction to monotone complexity, and \cite{PosP,PosNP} for examples of characterizations of the positive versions of the classes \textsc{P} and \textsc{NP}, in particular through extensions of first-order logic. The aforementioned paper \cite{AjtaiGurevich87}, which was the first to show the failure of Lyndon's theorem on finite structures, does so by reproving in particular an important result on monotone circuit complexity first proved in \cite{Circuits84}: $\text{Monotone-AC}^0\neq\text{Monotone}\cap\mathrm{AC}^0$.
\medskip

\noindent\textbf{Membership in subclasses of regular languages:} 

Related to our undecidability result, we can mention that there are syntactically defined classes of regular languages for which decidability of membership is an open problem. Such classes, also related to FO fragments, are the ones defined via quantifier-alternation: given a regular language, is it definable with an FO formula having at most $k$ quantifier alternations? Recent works obtained decidability results for this question, but only for the first $3$ levels of the quantifier alternation hierarchy \cite{PlaceZeitoun}. For higher levels, the problem remains open. Let us also mention the generalized star-height problem \cite{genstar}: can a given regular language be defined in an extended regular expression (with complement allowed) with no nesting of Kleene star? In this case it is not even known whether all regular languages can be defined in this way.
\medskip

\noindent\textbf{Quantitative extensions:}

First-order logic on words has been extended to quantitative settings, which naturally yields a negation-free syntax, because complementation becomes problematic in those settings.
This is the case in the theory of regular cost functions \cite{CostFun,KV12}, and in other quantitative extensions concerned with boundedness properties, such as MSO+U \cite{MSOU} or Magnitude MSO \cite{Magnitude}.
We hope that the present work can shed a light on these extensions as well.
\medskip

\subsection*{Notations and prerequisites}

If $i,j\in\N$, we note $[i,j]$ the set $\{i,i+1,\dots,j\}$. If $X$ is a set, we note $|X|$ its cardinal, and $\mathcal{P}(X)$ its powerset, i.e. the set of subsets of $X$.
We will note $A$ a finite alphabet throughout the paper. The set of finite words on $A$ is $A^*$.
The length of $u\in A^*$ is denoted $|u|$. If $L\subseteq A^*$ is a language, we will note $\overline{L}$ its complement.
We will note $\dom(u)=[0,|u|-1]$ the set of positions of a word $u$.
If $u$ is a word and $i\in\dom(u)$, we will note $u[i]$ the letter at position $i$, and $u[..i]$ the prefix of $u$ up to position $i$ included.
Similarly, $u[i..j]$ is the infix of $u$ from position $i$ to $j$ included and $u[i..]$ is the suffix of $u$ starting in position $i$ included. 
\medskip

We will assume that the reader is familiar with the notion of regular languages of finite words, and with some ways to define such languages: finite automata (DFA for deterministic and NFA for non-deterministic), finite monoids, and first-order logic.
See e.g. \cite{DG08} for an introduction to all the needed material.
\section{Monotonicity on words}

\subsection{Ordered alphabet}

In this paper we will consider that the finite alphabet $A$ is equipped with a partial order $\leq_A$.
This partial order is naturally extended to words componentwise: $a_1a_2\dots a_n\leq_A b_1b_2\dots b_m$ if $n=m$ and for all $i\in[1,n]$ we have $a_i\leq_A b_i$.

A special case that will be of interest here is when the alphabet is built as the powerset of a set $P$ of \emph{predicates}, i.e. $A=\mathcal P(P)$, and the order $\leq_A$ is inclusion. We will call this a \emph{powerset alphabet}.

Taking $A=\mathcal P(P)$ is standard in settings such as verification and model theory, where several predicates can be considered independently of each other in some position.

Powerset alphabets constitute a particular case of ordered alphabets. The results obtained in this paper are valid for both the powerset case and the general case. Due to the nature of the results (existence of a counter-example and undecidability result), it is enough to show them in the particular case of powerset alphabets to cover both cases. Moreover, the powerset alphabet case allows us to directly establish a link with Lyndon's theorem, which is stated in the framework of model theory. For these reasons, we will keep the more general notion of ordered alphabet for generic definitions, but we will prove our main results on powerset alphabets in order to directly obtain the stronger version of these results.

\subsection{Monotone languages}\label{subsec:clos}

We fix $A$ a finite ordered alphabet.

\begin{defi}
We say that a language $L\subseteq A^*$ is \emph{monotone} if for all $u\leq_A v$, if $u\in L$ then $v\in L$.
\end{defi}

\begin{exa}\label{ex:mon} Let $A=\{a,b\}$ with $a\leq_A b$.
Then $A^*bA^*$ is monotone but its complement $a^*$ is not monotone.
\end{exa}

\begin{defi}
Let $L\subseteq A^*$, the \emph{monotone closure} of $L$ is the language $\cl L=\{v\in A^*\mid\exists u\in L, u\leq_A v\}$. It is the smallest monotone language containing $L$.
\end{defi}

In particular, if $a\in A$, we will note $\cl a$ the set $\{b\in A \mid a\leq_A b\}$.
 
\begin{lem}\label{lem:clA} 
Given an NFA $\B$, we can compute in time $O(|\B|\cdot|A|)$ an NFA $\cl \B$ for the monotone closure of $L(\B)$.
\end{lem}
\begin{proof}
We build an NFA $\cl\B$ from $\B$, by replacing every transition $p\trans{a} q$ of $\B$ by $p\trans{\cl a}q$. We use here the standard convention where a transition $p\trans{X}q$ with $X\subseteq A$ stands for a set of transitions $\{p\trans{b}q\mid b\in X\}$. It is straightforward to verify that $\cl\B$ is an NFA for $\cl L$: any run of $\cl\B$ on some word $v$ can be mapped to a run of $\B$ on some $u\leq_A v$. 
\end{proof}

\begin{thm}
Given a regular language $L\subseteq A^*$, it is decidable whether $L$ is monotone. The problem is in \textsc{P} if $L$ is given by a DFA and \textsc{Pspace}-complete if $L$ is given by an NFA, on any alphabet with non-trivial order.
\end{thm} 
 
\begin{proof}
Notice that if $\B$ is an NFA, $L(\B)$ is monotone if and only if $L(\cl \B)\subseteq L(\B)$. This shows that the problem is in \textsc{Pspace} in general, and that it is in \textsc{P} when $\B$ is a DFA, since it reduces to checking emptiness of the intersection between $\cl \B$ and the complement of $\B$.
We show that the general problem is \textsc{Pspace}-hard by reducing from NFA universality. Let $\B$ be an NFA on a two-letter alphabet $A=\{a,b\}$. 
We build an NFA $\C$ of size polynomial in the size of $\B$ and recognizing $a\Ast+bL(\B)$, using standard NFA constructions. We now consider the monotonicity of $\C$ according to the alphabet order $a\leq_A b$. If $L(\C)$ is monotone, since for all $u\in\Ast$ we have $au\in L(\C)$, we obtain $bu\in L(\C)$ as well, so $u\in L(\B)$. So $L(\C)$ monotone implies $\B$ universal. Conversely, if $\B$ is universal, then $L(\C)=\Ast$ hence $L(\C)$ is monotone.
We have that $L(\B)=\Ast$ if and only if $L(\C)$ is monotone, thereby completing the \textsc{Pspace}-hardness reduction. This means that the monotonicity problem is \textsc{Pspace}-complete as soon as there are two comparable letters $a\leq_A b$ in the alphabet. Otherwise the problem is trivial, as any language is monotone on a trivially ordered alphabet.
\end{proof} 
 
\section{Positive first-order logic}\label{sec:FO}

\subsection{Syntax and semantics}\label{subsec:FO}

The main idea of positive FO, that we will note $\FOp$, is to guarantee via a syntactic restriction that it only defines monotone languages.

Notice that since monotone languages are not closed under complement (see \Cref{ex:mon}), we cannot allow negation in the syntax of $\FOp$. 
This means we have to add dual versions of classical operators of first-order logic.

This naturally yields the following syntax for $\FOp$:
$$\varphi,\psi:= \cl a(x)\mid x\leq y\mid x<y\mid \varphi\vee \psi\mid \varphi\wedge\psi \mid \exists x.\varphi\mid\forall x.\varphi$$

As usual, variables $x,y,\dots$ range over the positions of the input word. The semantics is the same as classical FO on words, with the notable exception that $\cl a(x)$ is true if and only if $x$ is labelled by some $b\in \cl a$. Unlike classical FO, it is not possible to require that a position is labelled by a specific letter $a$, except when $\cl a=\{a\}$. This is necessary to guarantee that only monotone languages can be defined.  
\medskip

\noindent\textbf{Formal semantics of $\FOp$} 

If $\varphi$ is a formula with free variables $\FV(\varphi)$, its semantics is a set $\semphi$ of pairs of the form $(u,\alpha)$, where $u\in \Ast$ and $\alpha:\FV(\varphi)\to\dom(u)$ a valuation for the free variables. We write indistinctively $u,\alpha\models\varphi$ or $(u,\alpha)\in\semphi$, to signify that $(u,\alpha)$ satisfies $\varphi$.
If $\FV(\varphi)=\emptyset$, we can simply write $u\models\varphi$ instead of $(u,\emptyset)\models\varphi$. In this case, the language recognized by $\varphi$ is $\semphi=\{u\in \Ast\mid u\models\varphi\}$.

We define $\semphi$ by induction on $\varphi$. 

\begin{itemize}
\item $u,\alpha\models {\cl a(x)}$ if $a\leq_A u[\alpha(x)]$.
\item $u,\alpha\models {x\leq y}$ if $\alpha(x)\leq \alpha(y)$.
\item $u,\alpha\models {x< y}$ if $\alpha(x)< \alpha(y)$.
\item $u,\alpha\models{\varphi\vee\psi}$ if $u,\alpha\models\varphi$ or $u,\alpha\models\psi$.
\item $u,\alpha\models{\varphi\wedge\psi}$ if $u,\alpha\models\varphi$ and $u,\alpha\models\psi$.
\item $u,\alpha\models {\exists x.\varphi}$ if there exists $i\in \dom(u)$ such that $u,\alpha[x\mapsto i]\models\varphi$.
\item $u,\alpha\models {\forall x.\varphi}$ if for all $i\in \dom(u)$, we have $u,\alpha[x\mapsto i]\models\varphi$.
\end{itemize}

Here the valuation $\alpha[x\mapsto i]$ maps $y$ to $\left\{\begin{array}{ll} i & \text{if } y=x\\\alpha(y) & \text{if } y\neq x \end{array}\right.$.

\begin{exa}On alphabet $A=\{a,b,c\}$ with $a\leq_A b$.
\begin{itemize}
\item $\forall x.\cl a(x)$ recognizes $\{a,b\}^*$.
\item $\exists x.\cl b(x)$ recognizes $\Ast b\Ast$.
\end{itemize}
\end{exa}

\begin{rem}
In the powerset alphabet framework where $A=\mathcal P(P)$, we can naturally view $\FOp$ as the negation-free fragment of first-order logic, by having atomic predicates $\cl a(x)$ range directly over $P$ instead of $A=\mathcal P(P)$. We can then drop the $\cl a$ notation, as predicates from $P$ are considered independently of each other. This way, $p(x)$ will be true if and only if the letter $S\in A$ labelling $x$ contains $p$. A letter predicate $\cl S(x)$ in the former syntax can then be expressed by $\bigwedge_{p\in S} p(x)$, so $\FOp$ based on predicates from $P$ is indeed equivalent to $\FOp$ based on $A$. We will take this convention when working on powerset alphabets.
\end{rem}

\begin{exa}
Let $A=\mathcal P(P)$ with $P=\{a,b\}$.
The formula $\exists x,y.~x\leq y\wedge a(x) \wedge b(y)$ recognizes $\Ast\{a,b\}\Ast+\Ast\{a\}\Ast\{b\}\Ast$.
\end{exa}

\subsection{Properties of $\FOp$}\label{subsec:prop}

\begin{lem}\label{lem:noorder}
Assume the order on $A$ is trivial, i.e. no two distinct letters are comparable. Then all languages are monotone, and any FO-definable language is $\FOp$-definable.
\end{lem}
\begin{proof}
The fact that all languages are monotone in this case follows from the fact that for two words $u,v$ we have $u\leq_ A v$ if and only if $u=v$.

If $L$ is definable by an FO formula $\varphi$, we can build an $\FOp$ formula $\psi$ from $\varphi$ by pushing negations to the leaves using the usual rewritings such as $\neg(\varphi\wedge\psi)=\neg\varphi\vee\neg\psi$ and $\neg(\exists x.\varphi)=\forall x.\neg\varphi$.  For all letter $a\in A$ and variable $x$, we then replace all occurrences of $\neg a(x)$ by $\bigvee_{b\neq a} b(x)$. Finally, the negation of $x\leq y$ (resp. $x<y$) can be written $y<x$ (resp. $y\leq x$).
\end{proof}

\begin{lem}\label{lem:FOclosed}
The logic $\FOp$ can only define monotone languages.
\end{lem}

\begin{proof}
This is done by induction on the $\FOp$ formula $\varphi$, where the induction property is strengthened to include possible free variables: for all $(u,\alpha)\in \semphi$ and $v\geq_A u$, we have $(v,\alpha)\in\semphi$.

\noindent\textbf{Base cases}:

Let $(u,\alpha)\in \sem{\cl a(x)}$ and $v\geq_A u$, we have $v[\alpha(x)]\geq_A u[\alpha(x)]\geq_A a$, so $(v,\alpha)\in \sem{\cl a(x)}$.

Let $(u,\alpha)\in \sem{x\leq y}$ and $v\geq_A u$.
We have $\alpha(x)\leq \alpha(y)$ so $(v,\alpha)\in \sem{x\leq y}$. The argument for $<$ instead of $\leq$ is identical.
\medskip

\noindent\textbf{Induction cases}:

Let $(u,\alpha)\in \sem{\varphi\vee\psi}$ and $v\geq_A u$. We have $(u,\alpha)\in\semphi$ or $(u,\alpha)\in\sempsi$. Therefore, by induction hypothesis, $(v,\alpha)\in\semphi$ or $(v,\alpha)\in\sempsi$, hence $(v,\alpha)\in \sem{\varphi\vee\psi}$. The argument for $\varphi\vee\psi$ is identical.

Let $(u,\alpha)\in \sem{\exists x.\varphi}$ and $v\geq_A u$. There exists $i\in\dom(u)$ such that $(u,\alpha[x\mapsto i])\in \semphi$. By induction hypothesis, $(v,\alpha[x\mapsto i])\in \semphi$. Hence, $(v,\alpha)\in \sem{\exists x.\varphi}$. The argument for $\forall$ is identical. \qedhere

\end{proof}

It is natural to ask whether the converse of \Cref{lem:FOclosed} holds: if a language is FO-definable and monotone, then is it necessarily $\FOp$-definable? This will be the purpose of \Cref{sec:K}.

\subsection{Ordered Ehrenfeucht-Fra\"issé games}\label{subsec:EF}

We will explain here how $\FOp$-definability can be captured by an ordered variant of Ehrenfeucht-Fra\"issé games, that we will call $\EF$-games.

This notion was defined in \cite{Stol95} for general structures, we will instantiate it here on finite words.

We define the $n$-round $\EF$-game on two words $u,v\in \Ast$, noted $\EFn(u,v)$.
This game is played between two players, Spoiler and Duplicator.

If $k\in\N$, a \emph{$k$-position} of the game is of the form $(u,\alpha,v,\beta)$, where $\alpha:[1,k]\to\dom(u)$ and $\beta:[1,k]\to \dom(v)$ are valuations for $k$ variables in $u$ and $v$ respectively.  We can think of $\alpha$ and $\beta$ as giving the position of $k$ previously placed tokens in $u$ and $v$.

A $k$-position $(u,\alpha,v,\beta)$ is \emph{valid} if for all $i\in[1,k]$, we have $u[\alpha(i)]\leq_A v[\beta(i)]$, and for all $i,j\in[1,k]$, $\alpha(i)\leq\alpha(j)$ if and only if $\beta(i)\leq\beta(j)$.

Notice the difference with usual EF-games: here we do not ask that tokens placed in the same round have same label, but that the label in $u$ is $\leq_A$-smaller than the label in $v$. This feature is intended to capture $\FOp$ instead of FO.

The game starts from the $0$-position $(u,\emptyset,v,\emptyset)$.

At each round, starting from a $k$-position $(u,\alpha,v,\beta)$, the game is played as follows.
If $k=n$, then Duplicator wins.
Otherwise, Spoiler chooses a position in one of the two words, and places token number $k+1$ on it.
Duplicator answers by placing token number $k+1$ on a position of the other word. Let us call $\alpha'$ and $\beta'$ the extensions of $\alpha$ and $\beta$ with these new tokens. If $(u,\alpha',v,\beta')$ is not a valid $(k+1)$-position, then Spoiler immediately wins the game, otherwise, the game moves to the next round with $(k+1)$-position $(u,\alpha',v,\beta')$.

We will note $u\preceq_n v$ when Duplicator has a winning strategy in $\EFn(u,v)$.

\begin{defi}
The quantifier rank of a formula $\varphi$, noted $\qr(\varphi)$ is its maximal number of nested quantifiers. It can be defined by induction in the following way: if $\varphi$ is atomic then $\qr(\varphi)=0$, otherwise, $\qr(\varphi\wedge\psi)=\qr(\varphi\vee\psi)=\max(\qr(\varphi),\qr(\psi))$ and $\qr(\exists x.\varphi)=\qr(\forall x.\varphi)=\qr(\varphi)+1$.
\end{defi}
The following Theorem shows the link between the $n$-round $\EF$ game and formulas of rank at most $n$.

\begin{thmC}[{\cite[Thm 2.4]{Stol95}}]\label{thm:EF}
We have $u\preceq_n v$ if and only if for all formulas $\varphi$ of $\FOp$ with $\qr(\varphi)\leq n$, we have $(u\models\varphi)\Rightarrow (v\models\varphi)$.
\end{thmC}

Since the proof of \Cref{thm:EF} does not appear in \cite{Stol95}, we will prove it in a general setting in Section \ref{sec:EFgen}.

Let us now see how we can use $\EF$ games to characterize $\FOp$-definability.

\begin{cor}\label{cor:EF} A language $L$ is not $\FOp$-definable if and only if for all $n\in\N$, there exists $(u,v)\in L\times \overline{L}$ such that $u\preceq_n v$.
\end{cor}

\begin{proof}

$\Leftarrow$ : Let $n\in\N$, there exists  $(u,v)\in L\times \overline{L}$ such that $u\preceq_n v$. By \Cref{thm:EF}, any formula of quantifier rank $n$ accepting $u$ must accept $v$, so no formula of quantifier rank $n$ recognizes $L$. This is true for all $n\in\N$, so $L$ is not $\FOp$-definable.

$\Rightarrow$ (contrapositive): Assume there exists $n\in\N$ such that for all $(u,v)\in L\times \overline{L}$, $u\not\preceq_n v$. By \Cref{thm:EF}, this means that for all $(u,v)\in L\times \overline{L}$, there exists a formula $\varphi_{u,v}$ of quantifier rank $n$ accepting $u$ but not $v$. Since there are finitely many $\FOp$ formulas of rank $n$ up to logical equivalence \cite[Lem 3.13]{Libkin04}, the set of formulas $F=\{\varphi_{u,v}\mid (u,v)\in L\times \overline{L}\}$ can be chosen finite. We define $\psi=\bigvee_{u\in L} \bigwedge_{v\notin L} \varphi_{u,v}$, where the conjunctions and disjunction are finite since $F$ is finite. For all $u\in L$, $u\models\bigwedge_{v\notin L} \varphi_{u,v}$ hence $u\models\psi$, and conversely, a word satisfying $\psi$ must satisfy some $\bigwedge_{v\notin L} \varphi_{u,v}$, so it cannot be in $\overline{L}$.
\end{proof}


\section{A counter-example language}\label{sec:K}

\subsection{The language $K$}

We will now answer the natural question posed in \Cref{subsec:prop}: is any FO-definable monotone language (on any ordered alphabet) also $\FOp$-definable?

This section is dedicated to the proof of the following Theorem:

\begin{thm}\label{thm:ce}
There is an FO-definable monotone language $K$ on a powerset alphabet that is not $\FOp$-definable.
\end{thm}

Let $P=\{a,b,c\}$ and $A=\mathcal P(P)$, ordered by inclusion.

We will note $\x,\y,\z$ for the letters $\{a,b\},\{b,c\},\{a,c\}$ respectively, and $\top$ for $\{a,b,c\}$.
If $x\in P$ we will often note $x$ instead of $\{x\}$ to lighten notations.

\begin{defi}\label{def:K}
We now define the desired language by: $$K := (\cl a\cl b\cl c)^* + \Ast\top\Ast.$$
\end{defi}

\noindent We claim that $K$ satisfies the requirements of \Cref{thm:ce}.

Notice that the second disjunct $\Ast\top\Ast$ could be omitted if we were to consider only the alphabet $A\setminus \{\top\}$. When sticking with a powerset alphabet, this disjunct is necessary to obtain an FO-definable language. Indeed, if we just define $K_0=(\cl a\cl b\cl c)^*$ on alphabet $A$, we have $K_0\cap (\top^*)= (\top\top\top)^*$. Since $\top^*$ is FO-definable but $(\top\top\top)^*$ is not (see \cite{DG08}), and FO-definable languages are closed under intersection, we have that $K_0$ is not FO-definable.

\subsection{FO-definability of $K$}

\begin{lem}\label{lem:KFO}
$K$ is monotone and FO-definable.
\end{lem}

\begin{proof}
The fact that $K$ is monotone is straightforward from its definition, as the union of two monotone languages.

We will show that $K$ is FO-definable in three different ways: using its minimal DFA, its syntactic monoid, and finally describing how an FO formula recognizing it can be defined.
This gives several complementary points of view, which can all be helpful for a deep understanding of this counter-example language.

Let us start with the automaton approach.
We use the classical characterizations of first-order definable languages \cite{DG08} by verifying that the minimal DFA $\A$ of $K$ is counter-free. That is, no word induces a non-trivial cycle in $\A$.

The minimal DFA $\A$ recognizing $K$ is depicted in \Cref{fig:aut}.  We note $\neg a=$\linebreak[4]$\{\emptyset, \{b\},\{c\},\{b,c\}\}$ the sub-alphabet of $A$ of letters not containing $a$, similarly for $\neg b$ and $\neg c$. The edges going to rejecting state $\bot$ are grayed and dashed, and the ones going to accepting sink state $q_\top$ are grayed, for readability. We also note $a'=\cl a\setminus\{\top\}=\{\{a\},\x,\z\}$, and similarly for $b',c'$.

\begin{figure}[h]
\centering
\scalebox{.9}{
\begin{tikzpicture}[shorten >=1pt,node distance=3cm,on grid,auto,initial text=,
every state/.style={inner sep=0pt,minimum size=6mm}]
    \node[state,accepting]	(p) {$q_a$};
        \node[state,below right=1.5cm and 3cm of p]	(q) {$q_b$};
        \node[state,above right=1.5cm and 3cm of q] (r) {$q_c$};
        \node[state,accepting, right=2cm of r] (top) {$q_\top$};
        \node[state, below=2cm of q] (bot) {$\bot$};
      \node[left=1cm of p] (start) {};
   \draw[->](start) -- (p);
   \path[->] 
   		(p) edge node {$a'$} (q)
   		(q) edge node[above left] {$b'$} (r)
   		(r) edge node[above] {$c'$} (p)
   		(p) edge[gray,bend left] node{$\top$} (top)
   		(q) edge[gray,bend right] node[gray,near start=.3]{$\top$} (top)
		(r) edge[gray] node[gray]{$\top$} (top)
		(top) edge[gray,loop above] node[gray]{$A$} ()
		(bot) edge[gray,bend right] node[gray,below right]{$\top$} (top)
		(p) edge[gray,dashed] node[gray,left] {$\neg a$} (bot)
   		(q) edge[gray,dashed] node[gray,left,near start=.3] {$\neg b$} (bot)
   		(r) edge[gray, dashed, bend left] node[gray] {$\neg c$} (bot)
   		(bot) edge[gray,loop below] node[gray]{$A\setminus\{\top\}$} ()
   	;
\end{tikzpicture}
}
\caption{The minimal DFA $\A$ of $K$}\label{fig:aut}
\end{figure}
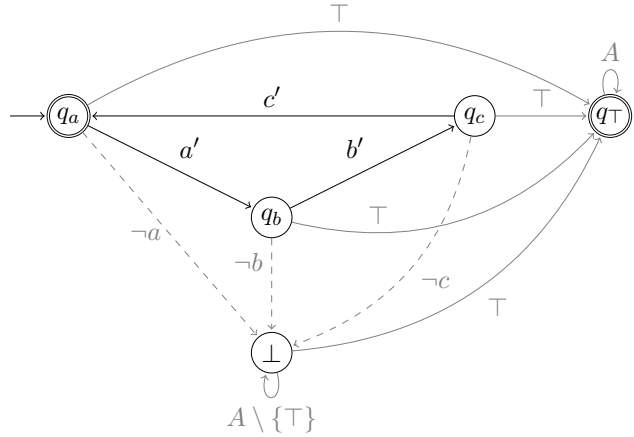

To show that $K$ is FO-definable, it suffices to show that $\A$ is counter-free, i.e. that there is no word $u\in\Ast$ , distinct states $p,q$ of $\A$, and integer $k$, such that $p\trans{u}q$ and $q\trans{u^k}p$. Assume for contradiction that such $u,p,q,k$ exist. Since the only non-trivial strongly connected component in $\A$ is $\{q_a,q_b,q_c\}$, these states are the only candidates for $p,q$. Since $p,q$ are distinct, it means $|u|$ is not a multiple of $3$, and $u$ induces a $3$-cycle, either $q_a\trans{u} q_b\trans{u} q_c\trans{u} q_a$ if $|u|\equiv 1\mod 3$ or in the reverse order if $|u|\equiv 2\mod 3$. Thus, the first letter of $u$ can be read from all states from $\{q_a,q_b,q_c\}$, while staying in this component. Such a letter does not exist, so we reach a contradiction. The DFA $\A$ is counter-free, so $K$ is FO-definable \cite{DG08}.
\end{proof}

\subsection{Syntactic monoid for the language $K$}\label{app:counter}

It is instructive to see what the syntactic monoid of $K$ looks like, in particular to get a first intuition on how an FO formula can be defined for $K$.

We depict this monoid $M$ in \Cref{fig:monoid}, using the eggbox representation based on Green's relations: boxes are $\mathcal J$-classes, lines are $\mathcal R$-classes, columns are $\mathcal L$-classes, and cells are $\mathcal H$-classes. See \cite{Green} for an introduction to Green's relations and eggbox representation.

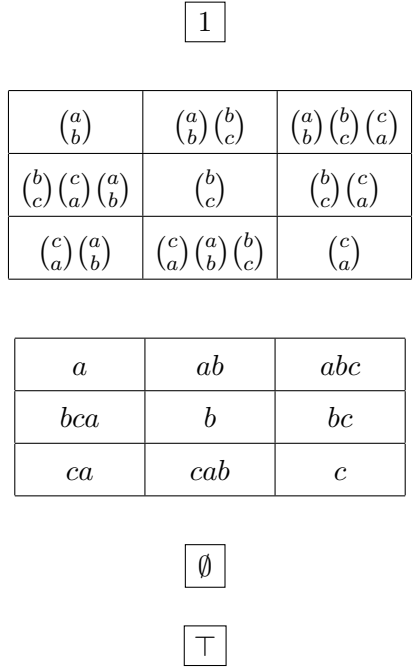
\begin{figure}[h]
\centering
\begin{tikzpicture}
\node[draw,rectangle, minimum size=15pt] (q1) {$1$};
\node[below=.5cm of q1] (q2) {\arraycolsep=4pt\def\arraystretch{1.8}
$\begin{array}{|c|c|c|}
\hline
\x & \x\y & \x\y\z \\
\hline
\y\z\x & \y & \y\z \\
\hline
\z\x & \z\x\y & \z \\
\hline
\end{array}
$};
\node[below=.5cm of q2] (q3) {\arraycolsep=17pt\def\arraystretch{1.5}
$\begin{array}{|c|c|c|}
\hline
a & ab & abc \\
\hline
bca & b & bc \\
\hline
ca & cab & c \\
\hline
\end{array}
$};
\node[draw,rectangle,below=.5cm of q3, minimum size=15pt] (q4) {$\emptyset$};
\node[draw,rectangle,below=.5cm of q4, minimum size=15pt] (q5) {$\top$};
\end{tikzpicture}
\caption{The syntactic monoid $M$ of $K$}\label{fig:monoid}
\end{figure}

The syntactic morphism $h:\Ast\to M$ is easily inferred, as elements of the monoid in $h(A)$ are directly named after the letter mapping to them. The accepting part of $M$ is $F=\{1,\x\y\z,\z\x\y,abc,\top\}$.

To show that $K$ is FO-definable, it suffices to verify that $M$ is aperiodic, which is directly visible on \Cref{fig:monoid}, as all $\mathcal H$-classes are singletons (see \cite{Green}).

\subsection{Defining an FO formula for the language $K$}\label{app:anchors}
We now give some intuition on how an FO-formula can recognize $K$.

Recall that $K=(\cl a\cl b\cl c)^* + \Ast\top\Ast.$
We describe here the behaviour of a formula witnessing that $K$ is FO-definable.

The $\Ast\top\Ast$ part of $K$ is just to rule out words containing $\top$ by accepting them, which can be done by a formula $\exists x.\top(x)$.
So we just need to design a formula $\varphi$ for $K'=(\cl a \cl b \cl c)^*\setminus (\Ast \top \Ast)$, assuming the letter $\top$ does not appear, the final formula will then be $\varphi\vee\exists x.\top(x)$. 

We will call \emph{forbidden pattern} any word that is not an infix of a word in $K'$.
Let us call \emph{anchor} a position $x$ such that either $x$ is labelled by a singleton, or $x$ is labelled by $\x$ (resp. $\y,\z$) with $x+1$ labelled by a letter different from $\y$ (resp. $\z,\x$).
The idea is that if $x$ is an anchor position of $u\in K'$, then there is only one possibility for the value of $x\mod 3$. If the first position is labelled by a letter from $\cl a$, we will consider that it is an anchor labelled $a$, otherwise we will reject the input word. Similarly, the last position is either a $c$ anchor or causes immediate rejection of the word.
If $x,y$ are successive anchor positions (i.e. with no other anchor positions between them), the word $u[x+1..y-1]$ is necessarily an infix of $(\x\y\z)^*$. We say that an anchor $x$ \emph{goes right-up} (resp. \emph{right-down}) if we can replace the letter $\duo\alpha\beta$ by $\alpha$ (resp. $\beta$) at position $x+1$ without having a forbidden pattern in the immediate neighbourhood of $x$. Notice that $x$ can not go both right-up and right-down. We define in the same way the left-up and left-down property by replacing $x+1$ with $x-1$.
For instance consider $u=\x\y\z\x\y c\x\y\z\x\y\y\z\x\y$, then apart from the first and last position there are two anchors: $x=5$ labelled $c$ and $y=10$ labelled $\y$, because it is followed by another $\y$. 


\begin{figure}[h]
\begin{center}
\includegraphics[scale=1.4]{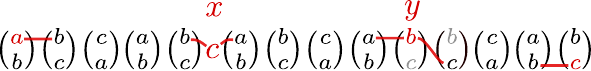}
\caption{A visualization of anchors}
\end{center}
\end{figure}

The anchor $x$ goes left-up and right-up, while the anchor $y$ goes left-up and right-down. If $d\in\{\text{up, down}\}$ is a direction, we say that two successive anchors $x<y$ \emph{agree on $d$} if $x$ goes right-$d$ and $y$ goes left-$d$. We say that $x$ and $y$ agree if they agree on some $d$.

Now, the formula $\varphi$ will express the following properties:
\begin{itemize}
\item for all $x,x+1$ consecutive anchors, the letters at positions $x,x+1, x+2$ do not form a forbidden pattern (omit $x+2$ if $x+1$ is the last position).
\item all non-consecutive successive anchors agree.
\end{itemize}
For instance the formula will accept the word $u$ above, as the anchors $0,x$ agree on up, $x,y$ agree on up, and $y,\last$ agree on down.

It is routine to verify that these properties can be expressed in FO, and that they indeed characterize the language $K'$. 

\subsection{Undefinability of $K$ in $\FOp$}

To prove that $K$ is the wanted counter-example, it remains to show:
\begin{lem}\label{lem:notFOp}
$K$ is not $\FOp$-definable.
\end{lem}
\begin{proof}
We establish this using \Cref{cor:EF}.
Let $n\in\N$, and $N=2^n$. We define $u=(abc)^{N}$ and $v=[\x\y\z]^{N-1}\x\y$.
Notice that $u\in K$, and $v\notin K$ because $|v|\equiv 2\;\mathrm{mod}\; 3$, and $v$ does not contain $\top$. By \Cref{cor:EF}, it suffices to prove that $u\preceq_n v$ to conclude.
We give a strategy for Duplicator in $\EFn(u,v)$. The strategy is an adaptation from the classical strategy showing that $(aa)^*$ is not FO-definable \cite{Libkin04}.
To simplify the description of the strategy, let us consider that prior to the game, tokens $\first$, $\last$ are placed on the first and last positions on $u$, and $\first',\last'$ on the first and last position of $v$. The strategy of Duplicator during the game is then as follows: every time Spoiler places a token in one of the words, Duplicator answers in the other by replicating the closest distance (and direction) to an existing token.
This strategy is illustrated in \Cref{fig:strat}, where move $i$ of Spoiler (resp. Duplicator) is represented by \Adam{$i$} (resp. \Eve{$i$}).

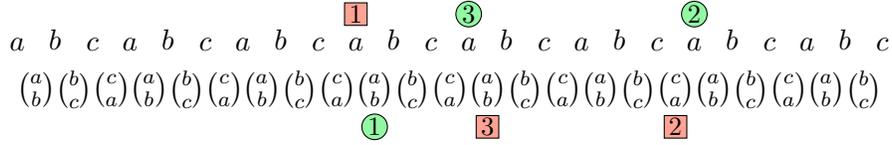
\begin{figure}[H]
\begin{center}
\scalebox{1}{
\begin{tikzpicture}[eve/.style={circle,draw=black, fill=evegreen!60, inner sep=0,minimum size=.34cm}, adam/.style={draw,fill=adamred!60,inner sep=0, minimum width=.3cm,minimum height=.3cm]}]
\def\s{.5}  
\foreach \k in {0,...,7}
{
\node at (\k*3*\s,0) {$a$};
\node at (\k*3*\s+\s,.06) {$b$};
\node at (\k*3*\s+2*\s,0) {$c$};
}
\foreach \k in {0,...,6}
{
\node at (\k*3*\s+.5*\s,-.6) {$\x$};
\node at (\k*3*\s+1.5*\s,-.6) {$\y$};
\node at (\k*3*\s+2.5*\s,-.6) {$\z$};
}
;
\node at (7*3*\s+.5*\s,-.6) {$\x$};
\node at (7*3*\s+1.5*\s,-.6) {$\y$};

\node[adam] at (3*3*\s,.4) {\small $1$};
\node[eve] at (3*3*\s+.5*\s,-1.1) {\small $1$};

\node[adam] at (5*3*\s+2.5*\s,-1.1) {\small $2$};
\node[eve] at (5*3*\s+3*\s,.4) {\small $2$};

\node[adam] at (4*3*\s+.5*\s,-1.1) {\small $3$};
\node[eve] at (4*3*\s,.4) {\small $3$};

\end{tikzpicture}
}
\caption{An example of Duplicator's strategy for $n=3$.}\label{fig:strat}
\end{center}
\end{figure}

Intuitively, the strategy of Duplicator is to match $u$ with the top row of $v$ if Spoiler plays close to the beginning of the words, and with the bottom row of $v$ if Spoiler plays close to the end.

We have to show that this strategy of Duplicator allows him to play $n$ rounds without losing the game. This proof is similar to the classical one for $(aa)^*$, see e.g. \cite{Libkin04}, and actually the strategy is exactly the same if we forget the letter labels. The main intuition is that the length of the non-matching intervals between $u$ and $v$ is at worst divided by $2$ at each round, and it starts with a length of $2^n$, so Duplicator can survive $n$ rounds.

Let us show that this strategy is indeed winning for Duplicator in $\EFn(u,v)$.

We will generally write $p,p'$ for related tokens, $p$ being the position in $u$ and $p'$ the position in $v$.

The proof works by showing that the following invariant holds: after $i$ rounds where Duplicator did not lose, if tokens in positions $p<q$ in $u$ are related to tokens $p'<q'$ in $v$, and $u[p..q]\not\leq_A v[p'..q']$, let us note $d=q-p, d'=q'-p'$; then $d=d'+1$ and $d\geq 2^{n-i}$.
In other words, if we call \emph{wrong interval} a factor $u[p..q]$ or $v[p'..q']$ such that $u[p..q]\not\leq_A v[p'..q']$, the invariant states that after $i$ rounds, the length of the smallest wrong interval in $u$ is at least $2^{n-i}$, and corresponding wrong intervals differ by $1$, the one in $u$ being longer. 
Before the first round, this invariant is true, as the only tokens are at the endpoints of $u$ and $v$, and we have $|u|=|v|+1$ and $|u|\geq 2^n$.
Now, assume the invariant true at round $i$, and consider round $i+1$. When Spoiler plays a token in one of the words, two cases can happen. If it is played between previous tokens $p,q$ (resp. $p',q'$) such that $u[p..q]\leq_A v[p'..q']$, then Duplicator will simply answer the corresponding position in the other word, and the smallest wrong interval is not affected.
If on the contrary, the new token is played in a minimal wrong interval, say $u[p,q]$ on position $r$, then Duplicator will answer by preserving the closest distance between $r-p$ and $q-r$. For instance if $r-p<q-r$, Duplicator will answer $r'=p'+(r-p)$. We can notice that by definition of the words $u$ and $v$, and since $u[p]\leq v[p']$ by the rules of the game, we have $u[p..r]\leq_A v[p'..r']$, and in particular $u[r]\leq_A v[r']$, so the move of Duplicator is legal. Moreover, since $q-r>r-p$, we have $q-r\geq\frac{q-p}2$, so using the induction hypothesis, $q-r\leq 2^{n-(i+1)}$. Moreover, since we had $(q-p)=(q'-p')+1$, we now have $(q-r)=(q-p)-(r-p)=(q'-p')+1-(r'-p')=(q'-r')+1$, so the invariant is preserved. The case where $r-p\geq q-r$ is symmetrical. If on the other hand Spoiler plays in $v$ a position $r'$ in a wrong interval $v[p'..q']$, then $\min(r'-p',q'-r')$ will be strictly smaller than $2^{n-(i+1)}$, and will be replicated by the answer $r$ of Duplicator in $u[p..q]$. This means that the new smallest wrong interval created in $u$ will have length at least $2^{n-(i+1)}$, thereby guaranteeing that the invariant is also preserved in this case. \qedhere

\end{proof}

\section{Lyndon's Theorem}
In this section we will see what the existence of this counter-example language $K$ means for Lyndon's theorem on other structures. We start by showing in Section \ref{subsec:struct} that it can be adapted to show the failure of Lyndon's theorem on finite structures. We then show in Section \ref{subsec:graphs} that the counter-example can also be encoded in finite directed graphs, and finally on undirected graphs.

\subsection{General structures}\label{subsec:struct}
We will consider here first-order logic on arbitrary signatures and unconstrained structures. All definitions of Section \ref{sec:FO} can be naturally extended to this general setting, and all results from Section \ref{subsec:prop} extend straightforwardly. We will later extend the $\EF$-game result as well.

Our goal is to see how \Cref{thm:ce} can be lifted to this general framework.

\begin{defi}
A formula $\varphi$ is \emph{monotone} in a predicate $P$ if whenever a structure $S$ is a model of $\varphi$, any structure $S'$ obtained from $S$ by adding tuples to $P$ is also a model of $\varphi$.
\end{defi}

\begin{exa}
On graphs, where the only predicate is the edge predicate, the formula asking for the existence of a triangle is monotone, but the formula stating that the graph is not a clique is not monotone.
\end{exa}

\begin{defi}
A formula $\varphi$ is \emph{positive} in $P$ if it never uses $P$ under a negation.
\end{defi}

Let us recall the statement of Lyndon's Theorem, which holds on general (possibly infinite) structures: 
\begin{thmC}[{\cite[Cor 2.1]{Lyndon59}}]\label{thm:Lyndon}
If $\psi$ is an FO formula monotone in predicates $P_1,\dots, P_n$, then it is equivalent to a formula positive in predicates $P_1,\dots,P_n$.
\end{thmC}

We will now see explicitly how the language $K$ from Section \ref{sec:K} can be used to show that Lyndon's Theorem fails on finite structures.

The failure of this theorem on finite structures was first shown in \cite{AjtaiGurevich87} with a very difficult proof, then reproved in \cite{Stol95} with a simpler one, using the Ehrenfeucht-Fra\"iss\'e technique. Still, the proof from \cite{Stol95} is quite involved compared to the one we present here.
\medskip

Since Lyndon's theorem can be found in the literature under different formulations, and since it is not clear at first sight that they are equivalent, we make it clear here how the construction of this paper applies to all of them. This also serves the purpose of making explicit the exact signature needed in each formulation to show the failure on finite structures with our method.

We will describe a signature by its sequence of arities, and add a symbol ${\uparrow}$ to specify monotone predicates. For instance $(2,1{\uparrow})$ describes a signature consisting of a binary predicate and a monotone unary predicate. Notice that the order is not important here, we are only interested in the multiset of pairs arity/monotonicity.

\medskip

\noindent(i) \textbf{Arbitrarily many monotone predicates}

This is the most general formulation of Lyndon's theorem, as made explicit in \Cref{thm:Lyndon}. Let us show that our language $K$ shows its failure on finite structures.

We will use here the fact that if $P=\{a,b,c\}$ is a set of monadic predicates, then a finite model over the signature $(\leq,a,b,c)$ where the order $\leq$ is total is simply a finite word on the powerset alphabet $A=\mathcal P(P)$.
Therefore, in order to view our words as general finite structures, it suffices to axiomatize the fact that $\leq$ a total order. This can be done with a formula $\psitot=(\forall x,y. ~x\leq y\vee y\leq x)\wedge (\forall x,y,z.~ x\leq y\wedge y\leq z\Rightarrow x\leq z)\wedge (\forall x,y.~x\leq y\wedge y\leq x\Rightarrow x=y)\wedge(\forall x.~x\leq x)$.
Notice that $\psitot$ is not monotone in the predicate $\leq$.

Let $\varphi$ be the FO-formula defining $K$, obtained in \Cref{lem:KFO}, and let $\psi=\varphi\wedge\psitot$.
Then, $\psi$ is monotone in predicates $a,b,c$, and finite structures on signature $(\leq,a,b,c)$ satisfying $\psi$ are exactly words of $K$. However, as we proved in \Cref{thm:ce}, no first-order formula that is positive in predicates $a,b,c$ can define the same class of structures, since the same formula interpreted on words would be an $\FOp$-formula for $K$.

This gives a counter-example on a signature $(2,1{\uparrow},1{\uparrow},1{\uparrow})$.

\medskip
\noindent(ii) \textbf{Single monotone predicate}

Other formulations of Lyndon's Theorem use a single monotone predicate, as in \cite{AjtaiGurevich87} and \cite{Stol95}. More precisely, \cite{AjtaiGurevich87} uses a signature $(1,2,3,3,3,3,3,4,1{\uparrow})$, and \cite{Stol95} a signature $(2,2{\uparrow})$.

We can encode the language $K$ in this framework, by using one binary predicate $A$ to represent all letter predicates. Let $K_3$ be $K$ restricted to words of length at least $3$. By \Cref{thm:ce} it is clear that $K_3$ is FO-definable but not $\FOp$-definable.

Let $\psi_3$ be an FO-formula stating that there are at least $3$ elements $0,1,2$, and that for all $y\notin\{0,1,2\}$ and for all $x$, $A(x,y)$ holds. We build the FO formula $\varphi'$ from the FO formula $\varphi$ recognizing the language $K_3$ by replacing every occurrence of $a(x)$ (resp. $b(x), c(x)$) by $A(x,0)$ (resp. $A(x,1), A(x,2)$).

\noindent Finally, we define the FO formula $\psi'=\psitot\wedge\psi_3\wedge\varphi'$. Finite structures on signature $(\leq,A)$ accepted by $\psi'$ are exactly those which encode words of $K_3$. No formula positive in $A$ can recognize this class of structures, otherwise we could obtain from it an $\FOp$-formula for $K_3$, by replacing every occurrence of $A(x,y)$ by $(a(x)\wedge y=0) \vee (b(x)\wedge y=1) \vee(c(x)\wedge y=2) \vee y\geq 3$.

We thus give a counter-example on a signature $(2,2{\uparrow})$, as was done in \cite{Stol95}.
\medskip

\noindent(iii) \textbf{Closure under surjective homomorphisms}

Lyndon's theorem is also often stated in the following way: if an FO formula defines a class of structures closed under surjective homomorphisms, then it is equivalent to a positive formula.
This formulation is equivalent to saying that the formula is monotone in all predicates. This case has been treated in \cite{RosenPHD}, using a slight modification of the construction from \cite{Stol95}, and building a counter-example on signature $(0,0,1{\uparrow},1{\uparrow},2{\uparrow},2{\uparrow},2{\uparrow})$. Notice that arities $0$ correspond to constants, which are always trivially monotone.

We can deal with this framework as well, by incorporating a predicate $\not\leq$ to the signature. Let $\psi_{\not\leq}$ be the formula obtained from $\psi$ defined in case (i) by pushing negations to the leaves and replacing all subformulas of the shape $\neg(x\leq y)$ with $x\not\leq y$. Let $$
\begin{array}{ll}
\psim &= \forall x,y.~(x\leq y\vee x\not\leq y)\\
\psip &=\exists x,y.~(x\leq y \wedge x\not\leq y)\\
\psi''&=(\psim\wedge\psi_{\not\leq} )\vee\psip.
\end{array}$$

Finite structures on the signature $(\leq,\not\leq,a,b,c)$ can be classified into three categories:
\begin{enumerate}
\item if there are $x,y$ such that $x\leq y \wedge x\not\leq y$, then the structure satisfies $\psip$ hence $\psi''$
\item otherwise, if there are $x,y$ such that $\neg(x\leq y\vee x\not\leq y)$, then the structure does not satisfy $\psim$ so it does not satisfy $\psi''$ either.
\item otherwise, $\not\leq$ is the complement of $\leq$, and the structure satisfies $\psi''$ if and only if it satisfies $\psi_{\not\leq}$.
\end{enumerate} 
Therefore, in $\psi_{\not\leq}$ we can use $\leq$ and $\not\leq$ freely, assuming that $\not\leq$ is actually the complement of $\leq$. In particular the $\psitot$ subformula of $\psi_{\not\leq}$ axiomatizes the fact that $\leq$ is a total order, provided that $\not\leq$ is its complement.
So the structures of item (3) are exactly the words of $K$, with an additional predicate $\not\leq$ which is the complement of $\leq$. Items (1),(2) guarantee that the class of finite structures accepted by $\psi''$ is monotone with respect to $\leq$ and $\not\leq$ as well.
As before, it is impossible to have a formula positive in all predicates accepting the same class of finite structures as $\psi''$, since replacing $x\not\leq y$ with $y<x$ in this formula would directly yield an $\FOp$-formula for $K$.

We thus give a counter-example on signature $(2{\uparrow},2{\uparrow},1{\uparrow},1{\uparrow},1{\uparrow})$, which can also be adapted to a signature $(2{\uparrow},2{\uparrow},2{\uparrow})$ by applying the trick of case (ii). Both results constitute an improvement upon the signature from \cite{RosenPHD}.

\subsection{$\EF$-games on arbitrary signature}\label{sec:EFgen}

\newcommand{\lmon}{{l}}
\newcommand{\larb}{{l'}}

We consider here $\FOp$ on an arbitrary signature consisting of $\lmon$ monotone predicates $P_1,\dots, P_\lmon$ of arity $r_1,\dots,r_\lmon$ respectively, and $\larb$ non-monotone predicates $R_1,\dots,R_\larb$ of arity $r'_1,\dots, r'_\larb$ respectively. 

Thus the syntax of $\FOp$ in this setting is:
$$\varphi,\psi:= P_i(x)\mid R_i(x)\mid \neg R_i(x)\mid \varphi\vee \psi\mid \varphi\wedge\psi \mid \exists x.\varphi\mid\forall x.\varphi$$

Notice that we allow the negation of predicates from $R_i$.

We can additionally assume that these formulas will only be evaluated on structures verifying certain axioms, for instance on structures where a predicate $\leq$ evaluates to a linear order, or where the predicate $=$ corresponds to equality. These will be called $\sigma$-structures in the following, where the signature $\sigma$ can be enriched by such axioms.


The $\EF$-game on two $\sigma$-structures $u,v$ is defined as before, with the following generalization:
at a given stage $(u,\alpha,v,\beta)$ of the game, where $\alpha$ (resp. $\beta$) is a valuation in $u$ (resp. $v$) for already played tokens, Duplicator must ensure:
\begin{itemize}
\item For any monotone predicate $P_i$ and tuple $\vec x$ of $r_i$ played tokens, we must have $u,\alpha\models P_i(\vec x)\Rightarrow v,\beta\models P_i(\vec x)$.
\item For any non-monotone predicate $R_i$ and tuple $\vec x$ of $r'_i$ played tokens, we must have $u,\alpha\models P_i(\vec x)\Leftrightarrow v,\beta\models P_i(\vec x)$.
\end{itemize}

As before, we note $u\preceq_n v$ if Duplicator has a strategy to win the $n$-round $\EF$-game between $s$ and $t$.
We want to show the following theorem, generalizing \Cref{thm:EF}, and formulated in \cite{Stol95}:

\begin{thmC}[{\cite[Thm 2.4]{Stol95}}]\label{thm:EFgen}
In a general setting, for any $\sigma$-structures $u,v$, we have $u\preceq_n v$ if and only if for all formulas $\varphi$ of $\FOp$ with $\qr(\varphi)\leq n$, we have $(s\models\varphi)\Rightarrow (t\models\varphi)$.
\end{thmC}

The proof is an adaptation of the classical proof for correctness of EF-games, see e.g. \cite{Libkin04}.

Since $\FOp$ is a fragment of FO, we can directly use the following Lemma:
\begin{lemC}[{\cite[Lem 3.13]{Libkin04}\label{lem:finrank}}]
Let $n,k\in \N$. Up to logical equivalence, there are finitely many formulas of quantifier rank at most $n$ using $k$ free variables.
\end{lemC}

We will now show a strengthening of \Cref{thm:EFgen}, where free variables are incorporated in the statement:

\begin{thm}
Let $n,k\in\N$, $u,v$ $\sigma$-structures, $\alpha:[1,k]\to\dom(u)$ and $\beta:[1,k]\to\dom(v)$ be valuations for $k$ variables $x_1,\dots,x_k$ in $u,v$ respectively.
Then Duplicator wins $\EFn(u,\alpha,v,\beta)$ if and only if for any $\FOp$ formula $\varphi$ with $\qr(\varphi)\leq n$ using $k$ free variables $x_1\dots x_k$, we have $u,\alpha\models\varphi~\Rightarrow~v,\beta\models\varphi$.
\end{thm}

\begin{proof}
We prove this by induction on $n$.

\noindent\textbf{Base case $n=0$}:

Notice that quantifier-free formulas of $\FOp$ are just positive boolean combinations of atomic formulas of the form $P_i(\vec x)$, $R_i(\vec x)$ or $\neg R_i(\vec x)$. We will note $Q(\vec x)$ for such an arbitrary atomic formula.
Let $\varphi$ be such a formula with $k$ free variables accepting $u,\alpha$ but rejecting $v,\beta$. This happens if and only if there is an atomic formula $Q(\vec x)$ such that $u,\alpha\models Q(\vec x)$ and $v,\beta\not\models Q(\vec x)$. This is equivalent to saying that $(u,\alpha,v,\beta)$ is not a valid $k$-position, i.e. Spoiler wins the $0$-round game $\EF_0(u,\alpha,v,\beta)$.
\medskip

\noindent\textbf{Induction case}:
Assume there is an $\FOp$ formula $\varphi$ with $\qr(\varphi)\leq n$, accepting $u,\alpha$ but not $v,\beta$.
The formula $\varphi$ is a positive combination of atomic formulas, formulas of the form $\exists x.\psi$, and formulas of the form $\forall x.\psi$. Therefore, one of these formulas accepts $u,\alpha$ but not $v,\beta$. If it is an atomic formula, then Spoiler immediately wins $\EFn(u,\alpha,v,\beta)$ as in the base case.

If it is a formula of the form $\exists x.\psi$, then Spoiler can use the following strategy: pick a position $p$ witnessing that the formula is true for $u,\alpha$, and play the position $p$ in $u$. Duplicator will answer a position $p'$ in $v$, and the game will move to $(u,\alpha',v,\beta')$, where $\alpha'=\alpha[x\mapsto p]$ and $\beta'=\beta[x\mapsto p']$. Since the formula $\psi$ has quantifier rank at most $n-1$, and accepts $u,\alpha'$ but not $v,\beta'$, by induction hypothesis Spoiler can win in the remaining $n-1$ rounds of the game.

Now if it is a formula of the form $\forall x.\psi$, then Spoiler can do the following: pick a position $p'$ witnessing that the formula is false for $v,\beta$, and play the position $p'$ in $v$. Duplicator will answer a position $p$ in $u$, and the game will move to $(u,\alpha',v,\beta')$, where $\alpha'=\alpha[x\mapsto p]$ and $\beta'=\beta[x\mapsto p']$. Since the formula $\psi$ has quantifier rank at most $n-1$, and accepts $u,\alpha'$ but not $v,\beta'$, by induction hypothesis Spoiler can win in the remaining $n-1$ rounds of the game.
\medskip

Let us now show the converse implication. We assume any formula of quantifier rank at most $n$ accepting $u,\alpha$ must accept $v,\beta$, and we give a strategy for Duplicator in $\EFn(u,\alpha,v,\beta)$.

Suppose Spoiler places token $x$ at position $p$ in $u$. Let $\alpha'=\alpha[x\mapsto p]$.
By \Cref{lem:finrank}, up to logical equivalence, there is only a finite set $F$ of $\FOp$ formulas of rank at most $n-1$ with $k+1$ free variables accepting $u, \alpha'$. Let $\psi=\bigwedge_{\varphi\in F} \varphi$.
Then $u,\alpha$ satisfies the formula $\exists x.\psi$ of rank $n$ (as witnessed by $p$), so by assumption we also have $v,\beta\models\exists x.\psi$. This means there is a position $p'$ of $v$ such that $v,\beta'\models \psi$, where $\beta'=\beta[x\mapsto p']$. Duplicator can answer position $p'$ in $v$, and by induction hypothesis he will win the remaining of the game, since every formula of $F$ accepts $v,\beta'$.
\smallskip

Suppose now that Spoiler places token $x$ at position $p'$ in $v$. Let $\beta'=\beta[x\mapsto p']$.
Let $F$ be the finite set of formulas (up to equivalence) of quantifier rank at most $n-1$ and with $k+1$ free variables, that reject $v,\beta'$. Let $\psi=\bigvee_{\varphi\in F}\varphi$, and $\psi'=\forall x.\psi$. By construction, $x=p'$ witnesses that $\psi'$ does not accept $v,\beta$. Our assumption implies that it does not accept $u,\alpha$ either. So there is $p\in\dom(u)$ such that $u,\alpha'\not\models \forall x.\psi$, where $\alpha'=\alpha[x\mapsto p]$.
Duplicator can answer position $p$ in $u$.  If a formula $\varphi$ of rank at most $n-1$ is true in $u,\alpha'$, then by construction it cannot appear in $F$, therefore it is also true in $v,\beta'$. By induction hypothesis, Duplicator wins the remaining $(n-1)$-round game starting from $(u,\alpha',v,\beta')$.
\end{proof}

This achieves the proof of \Cref{thm:EFgen}, and its instantiation \Cref{thm:EF} on finite words. The proof of \Cref{cor:EF} is exactly identical in this general setting, so $\EF$ game can be used to prove that some properties of general structures are not expressible in $\FOp$:

\begin{cor}\label{cor:EFgen}
A class of $\sigma$-structures $C$ is not $\FOp$-definable if and only if for all $n\in\N$, there exists $(u,v)\in C\times \overline{C}$ such that $u\preceq_n v$.
\end{cor}

\subsection{Finite directed graphs}\label{subsec:graphs}

Our goal is now to show that Lyndon's theorem fails on finite directed graphs, i.e. on finite structures where the signature consists in one (monotone) binary predicate, in addition to (non-monotone) equality. To our knowledge this is a new result.

The positive FO formulas on graphs, that we will call again $\FOp$, is defined via the following syntax:
$$\varphi,\psi:= E(x,y)\mid x=y\mid x\neq y \mid \varphi\vee \psi\mid \varphi\wedge\psi \mid \exists x.\varphi\mid\forall x.\varphi$$

while general FO formulas can additionally use predicates of the form $\neg E(x,y)$. 
A class $\C$ of graphs is \emph{monotone} if whenever $G\in \C$, and $G'$ is obtained from $G$ by adding edges, then $G'\in\C$. It is straightforward to adapt the proof of \Cref{lem:FOclosed} to show that $\FOp$ can only define monotone classes of graphs.

Notice that equality/inequality predicates were not needed in the case of words since this was expressible with the order predicates $\leq$ and $<$.

The goal of this section is to prove the following result:
\begin{thm}[Failure of Lyndon's Theorem on finite directed graphs]\label{thm:lyndongraphs}
There exists an FO-definable monotone class of directed graphs, which is not $\FOp$-definable.
\end{thm}

Let us start by giving an informal proof sketch. Our goal will simply be to encode the language $K$ from Section \ref{sec:K} as a set of graphs. Thus, the graphs of interest will have a very specific shape, allowing to encode words on alphabet $A=\mathcal P(\{a,b,c\})$. In order to ensure monotonicity, instead of forbidding all patterns that break the encoding, we will instead accept any graph having ``too many edges''. This is the same idea as in the ``Closure under surjective homomorphisms'' paragraph of Section \ref{subsec:struct}.
Thus we will have two kinds of constraints:
\begin{itemize}
\item A formula $\psim$ asking for some edges to be present, i.e. rejecting graphs that cannot encode a word because of a lack of edges.
\item A formula $\psip$ that will accept any graph falling outside of the required shape of an encoding, because of an excess of edges.
\end{itemize}

Both $\psim$ and $\psip$ will be $\FOp$ formulas, thus describing monotone classes of graphs. The graphs encoding words of $A^*$ will be the models of $\psim$ that are not models of $\psip$. Let us call $\Gw$ the set such graphs, encoding words of $A^*$. We call $\GK$ the subset of $\Gw$ consisting of graphs encoding words of $K$. Our monotone language of graphs witnessing failure of Lyndon's theorem will be $\GK\cup \sem{\psip}$.
Let $\phiKG$ be an FO formula accepting $\GK$ among graphs from $\Gw$, the behaviour of $\phiKG$ being irrelevant outside of $\Gw$. This formula $\phiKG$ will be obtained from the FO formula $\phiK$ for the language $K$, by interpreting predicates $a(x)$, $b(x)$, $c(x)$, and $x\leq y$ in our encoding. For technical reasons, our encoding will also use some distinguished vertices $\vec x$, shared by all formulas, and existentially quantified. Thus our final formula will be of the form $\phi:=\exists \vec x.\psim\wedge(\phiKG\vee\psip)$. It accepts a monotone language of graphs: $\GK\cup \sem{\psip}$.
In order to show that there is no positive formula equivalent to $\phi$, we will replicate the $\EF$ game of \Cref{lem:notFOp} using graphs from $\Gw$.
\medskip

Let us now move to the detailed construction. Let $E(x,y)$ be the edge predicate, that we will often note $x\to y$ for simplicity. Similarly, we will note $x\to y\to z$ as a shortcut for $E(x,y)\wedge E(y,z)$. This arrow will not be at risk of being confused with an implication symbol, since we will avoid the use of implication to obtain negation-free formulas.
In the following, we assume three vertices $x_a,x_b,x_c$ are pointed in the graph. We will call them \emph{sources}, represented by circles in figures. We will call ``squares'' the vertices other than $x_a,x_b,x_c$. We will impose that the subgraph induced by the sources is a particular one, the only purpose of this is to be able to uniquely identify these three vertices.
We define $\Gw$ to be the set of graphs satisfying the following properties:
\begin{itemize}[align=left]
\item[(sources)] $x_a,x_b,x_c$ are distinct, and they induce the following subgraph:
\begin{minipage}{5cm}
\scalebox{.8}{
\begin{tikzpicture}[shorten >=1pt,node distance=1.5cm,on grid,auto,
rect/.style={rectangle,draw,minimum size=6mm},
circ/.style={circle,draw,minimum size=4mm},
pent/.style={regular polygon,regular polygon sides=5,draw,minimum size=4mm}]
    \node[circ] (a) {$x_a$};
	\node[circ,right=of a] (b) {$x_b$};
	\node[circ,right=of b] (c) {$x_c$};

	\path[->] 
	(a) edge (b)
	(b) edge[bend left] (c)
	(c) edge[bend left] (b)
	;	 
\end{tikzpicture}
}
\end{minipage}
\item[(in-edge)] $x_a$ is the only vertex with no in-edge
\item[(cycle)] There is no cycle of length at most $3$ other than the $2$-cycle on $x_b,x_c$.
\item[(order)]  Any two squares are related by an edge.
\item[(direction)] There is no edge from a square to a source.
\end{itemize}
The rule (direction) is actually optional for the correctness of the construction, but it simplifies the exposition.
The next lemma justifies the choice of name for the rule (order):
\begin{lem} If $G$ is a graph in $\Gw$, the edge relation defines a strict total order on squares.
\end{lem}
\begin{proof}Since $3$-cycles are forbidden, and any two squares are related, the edge relation is transitive on squares: for any $x\to y \to z$, there is an edge $x\to z$. Since self-loops and $2$-cycles are forbidden, the relation is irreflexive and antisymmetric. Therefore, it defines a strict total order on squares.
\end{proof}
\Cref{fig:Gw} shows the shape of such a graph with four squares. Notice that the edges from sources to squares can be arbitrary, except that there must be an edge from a source to the first square in the order, because of rule (in-edge).
\begin{figure}[H]
\centering
\scalebox{.8}{
\begin{tikzpicture}[shorten >=1pt,node distance=2cm,on grid,auto,
rect/.style={rectangle,draw,minimum size=6mm},
circ/.style={circle,draw,minimum size=4mm},
pent/.style={regular polygon,regular polygon sides=5,draw,minimum size=4mm}]
    \node[circ] (a) {$x_a$};
	\node[circ,right=of a] (b) {$x_b$};
	\node[circ,right=of b] (c) {$x_c$};

	\node[rect, below left=1.7cm and 1cm of a] (s1) {};
	\node[rect, right=2cm of s1] (s2) {};
	\node[rect, right=2cm of s2] (s3) {};
	\node[rect, right=2cm of s3] (s4) {};
	\path[->] 
	(a) edge (b)
	(b) edge[bend left] (c)
	(c) edge[bend left] (b)
	(s1) edge (s2)
		 edge[bend right] (s3)
		 edge[bend right] (s4)
	(s2) edge (s3)
		edge[bend right] (s4)
	(s3) edge (s4)
;	 
	\path[->]
	(a) edge (s1)
		edge (s3)
	(b) edge (s2)
		edge (s3)
	(c) edge (s4);
\end{tikzpicture}
}
\caption{A graph of $\Gw$}\label{fig:Gw}
\end{figure}
Let us give at this stage more details about sources: their role will be to encode the unary predicates $a,b,c$ on words, and the constraints of $\Gw$ allow to identify them without ambiguity in an unlabeled graph: $x_a$ is the only vertex with no in-edge, and $x_b,x_c$ form the only $2$-cycle, $x_b$ being the vertex connected to $x_a$.

We now need to express all these constraints via $\FOp$ formulas $\psim$ and $\psip$. Recall that formulas in $\psip$ are meant to express when a graph is not in $\Gw$ because of extra edges.
Let $\ES=\{(x_a,x_b),(x_b,x_c), (x_c,x_b)\}$ the set of edges in the induced graph of rule (source), and $\overline{\ES}=\{x_a,x_b,x_c\}^2\setminus \ES$ its complement.
We will also use as auxiliary formulas:
\begin{itemize}
\item $\Ci(x):=(x=x_a)\vee (x=x_b)\vee (x=x_c)$, stating that $x$ is a source,
\item $\Sq(x):=(x\neq x_a)\wedge (x\neq x_b)\wedge (x\neq x_c)$ stating that $x$ is a square, 
\item $\Sq_a(x):=(x\neq x_b)\wedge (x\neq x_c)$ for squares or $x_a$.
\end{itemize}

$$\arraycolsep=3pt\def\$arraystretch{1.3}
\begin{array}{|c|c|c|}
\hline
\textbf{Constraint} & \bm{\psim} & \bm{\psip} \\
\hline
\text{(sources)} & \bigwedge_{(x,y)\in \ES} (x\to y\wedge x\neq y) & ~~\bigvee_{(x,y)\in \overline{\ES}} x\to y \\
\hline
\text{(in-edge)} & \forall y. (y=x_a\vee\exists x. x\to y) & ~~\exists x. x\to x_a~~\\
\hline
\text{(cycle)} & & \begin{array}{ll}&\exists x.x\to x\\
	\vee&\exists x,y.(\Sq_a(x)\vee \Sq_a(y))\wedge (x\to y\to x)\\
		 \vee &\exists x,y,z.( x\to y\to z\to x)~~\\
		 \end{array}\\
\hline
\text{(order)} & \begin{array}{ll}\forall x,y.& \Ci(x)\vee \Ci(y)\vee(x=y)\\& \vee (x\to y)\vee (y\to x)\end{array} & \\
\hline
\text{(direction)} & & \exists x,y. \Sq(x)\wedge \Ci(y)\wedge x\to y \\
\hline
\end{array}$$
\smallskip

We finally take for the $\psim$ (resp. $\psip$) the conjunction (resp. disjunction) of the formulas in its column.
We obtain that by definition, a graph with marked vertices $x_a,x_b,x_c$ is in $\Gw$ if and only if it satisfies $\psim$ but not $\psip$.

It remains to describe how words on alphabet $A=\mathcal P(\{a,b,c\})$ can be encoded into these graphs.
Let $G\in\Gw$, we will associate to it a word $u$ using the following rules:
\begin{itemize}
\item The positions $\dom(u)$ of $u$ correspond to the square vertices of $G$.
\item The order $<$ on $\dom(u)$ corresponds to edges between square vertices.
\item If $x\in\dom(u)$, we will say that $a(x)$ (resp. $b(x),c(x)$) is true if there is in $G$ an edge $x_a\to x$ (resp. with $x_b,x_c$).
\end{itemize}

With this, we can see that the graph of \Cref{fig:Gw} encodes the word $ab\x c$.

Notice that this encoding is actually a bijection between graphs of $\Gw$ and words in $A^*$ that do not start with letter $\emptyset$.

We note $\GK$ the graphs of $\Gw$ that encode a word $u\in K$, where $K$ is the language from Section \ref{sec:K}. Recall that letter $\emptyset$ is never allowed in a word of $K$.
Let $\phiK$ be the FO-formula for $K$ (from \Cref{lem:KFO}), on signature $(\leq,a,b,c)$. We want to build a formula $\phiKG$, recognizing the graphs of $\GK$ among those of $\Gw$. The idea is to restrict quantification to square vertices, and replace atomic predicates by their graph interpretations. This is done by induction on formulas:

\begin{itemize}
\item $(x\leq y)^G:=E(x,y)\vee x=y$
\item $(x<y)^G:=E(x,y)$
\item $(\alpha(x))^G:= E(x_\alpha,x)$, for $\alpha\in\{a,b,c\}$.
\item $(\exists x.\varphi)^G:=\exists x.\Sq(x)\wedge\varphi^G$
\item $(\forall x.\varphi)^G:=\forall x.\Ci(x)\vee\varphi^G$
\item $(\varphi\wedge\psi)^G:=\varphi^G\wedge\psi^G$
\item $(\varphi\vee\psi)^G:=\varphi^G\vee\psi^G$
\item $(\neg \varphi)^G:=\neg(\varphi^G)$
\end{itemize}

Notice that since $\phiK$ is not syntactically positive, the formula $\phiKG$ will contain negations as well.

Finally, let us define $$\phi:=\exists x_a,x_b,x_c.(\psim\wedge(\phiKG\vee\psip)).$$

\begin{rem}
Contrarily to the last construction of Section \ref{subsec:struct}, how to parenthesize is important here, because of the existential quantification on sources $x_a,x_b,x_c$.
Indeed, we must avoid being too permissive by allowing a choice of sources that would validate $\psip$ without validating $\psim$, thereby accepting some unwanted graphs because of a bad choice of sources. Actually, with the other choice of parentheses $(\psim\wedge\phiKG)\vee\psip$, we would have $\phi$ accepting any graph containing an edge, because of rule (in-edge) in $\psip$. 
\end{rem}

\begin{lem}
The formula $\phi$ is monotone, and $\sem{\phi}\cap\Gw=\GK$.
\end{lem}
\begin{proof}
We first show that $\sem{\phi}\cap\Gw=\GK$. Remark that given a graph in $\Gw$, there is only one possible choice of $x_a,x_b,x_c$ that validates $\psim$: they are the only vertices forming the subgraph given in rule (sources). Thus, after $x_a,x_b,x_c$ have been fixed, the fact that the graphs of $\Gw$ accepted by $\phi$ are exactly those from $\GK$ follows from this: the formula $\phiKG$ interprets correctly a graph of $\Gw$ as a word of $\Ast$, and accepts it if and only if this word is in $K$. This fact is simply the correctness of the $(\cdot)^G$ transformation, proven by straightforward induction.

We now move to the monotonicity property: let $G\in\sem{\phi}$, and $G'$ be obtained from $G$ by adding some edges.

First, if $G\in\sem{\psip}$, then so does $G'$, because $\psip$ is positive.
If on the contrary $G\notin\sem{\psip}$, this means $G\in\GK$, so $G$ encodes a word $u\in K$. Then two cases can occur for $G'$:
\begin{itemize}
\item If $G'\in\Gw$, then it encodes a word $u'\geq_A u$, because adding edges to a graph of $\Gw$ while staying inside $\Gw$ translates into adding letter predicates in the corresponding word. Since $K$ is monotone, we have $u'\in K$, so $G'\in\GK$.
\item If $G'\notin \Gw$, recall that $\Gw=\sem{\psim}\setminus\sem{\psip}$. Since $G\in\sem{\psim}$ and $\psim$ is positive, we have $G'\in\sem{\psim}$. So we can conclude $G'\in\sem{\psip}$.
\end{itemize}
In all cases, we obtain $G'\in\sem{\phi}$, so $\phi$ is monotone.
\end{proof}

\begin{lem}
There is no $\FOp$ formula recognizing $\sem{\phi}$.
\end{lem}
\begin{proof}
We use $\EF$-games on graphs of $\Gw$. By \Cref{cor:EFgen}, we can use $\EF$-games on signature $(E,=)$, with $E$ monotone and $=$ non-monotone and correspond to equality, to show that there is no positive formula defining $\sem{\phi}$. 
We will replicate the game from \Cref{lem:notFOp}, using graphs instead of words. More precisely, let $n\in\N$ and $N=2^{n}$. Let $u=(abc)^N$ and $v=[\x\y\z]^{N-1}\x\y$, recall that $u\in K$ and $v\notin K$.
We associate to $u$ and $v$ graphs $G_u$ and $G_v$ in $\Gw$ according to the above encoding.
We have $G_u\in\GK$ and $G_v\in\Gw\setminus\GK$, so $G_u\in\sem{\phi}$ and $G_v\notin\sem{\phi}$. It now suffices to show that Duplicator can win $\EF_n(G_u,G_v)$. His strategy can actually mimic exactly the strategy $\sigma_D$ for Duplicator in $\EF_{n}(u,v)$. It suffices to play as $\sigma_D$ on squares, and if Spoiler plays a source in one of the graphs, Duplicator answers with the corresponding source in the other. It is straightforward to verify that this is a winning strategy for Duplicator in $\EF_n(G_u,G_v)$, as both order and letter predicates on words directly translate to edge predicates on graphs of $\Gw$.
This shows that there is no $\FOp$ formula equivalent to $\phi$.
\end{proof}

This concludes the proof of \Cref{thm:lyndongraphs}.

\subsection{Finite undirected graphs}\label{subsec:undirgraphs}

We describe in this section how to lift the counter-example from directed graphs to undirected ones.

As the proof scheme follows the same pattern as in the directed case, we will go in less details and just describe here the modifications to be made to lift the previous proof to undirected graphs. 

We will again use distinguished ``source vertices'' that will encode letter predicates, except that now there will be more than one source per letter.
As we can no longer use the orientation of edges to isolate these sources without ambiguity, we will instead use cycles, in the same spirit as the $x_b-x_c$ cycle in the directed case: the sources will form the only cycles of length at most $5$. More precisely, the $a$-sources will form a $3$-cycle, the $b$-sources a $4$-cycle, and the $c$-sources a $5$-cycle, for a total of $12$ sources. Moreover, there are no other edges between sources than the ones forming these cycles. This means that once again we completely impose the graph induced by the $12$ sources: it has to consist in three disjoint cycles of size $3,4,5$. We can therefore assume that we have a formula $\Ci(x)$ stating that $x$ is a source, and formulas $\Ci_a(x),\Ci_b(x),\Ci_c(x)$ specifying the letter of this source. Since sources will again be explicitly quantified in a formula, it is still possible to state in $\FOp$ that a vertex $x$ is not a source, via a formula $\neg\Ci(x)$ simply asserting that $x$ is different from all sources.

As before, we will use vertices called ``squares'' to encode positions of the word. This time, squares will not be all non-source vertices, but will be defined as follows: a square is any non-source vertex that is connected to a source by an edge. This can be expressed by a formula $\Sq(x)=\neg\Ci(x)\wedge\exists y.\Ci(y)\wedge E(x,y)$. We still need to encode the total order on squares, but since edges are not oriented anymore, we will make use of oriented ``meta-edges'' as described by the following picture:

\begin{figure}[H]
\centering
\scalebox{.8}{
\begin{tikzpicture}[shorten >=1pt,node distance=1.5cm,on grid,auto,
rect/.style={rectangle,draw,minimum size=6mm},
diam/.style={diamond,draw,minimum size=4mm}]
    \node[rect] (a) {$x$};
	\node[diam,right=of a] (b) {};
	\node[diam,right=of b] (c) {};
	\node[diam,right=of c] (d) {};
	\node[rect,right=of d] (e) {$y$};
	\node[diam,above=of b] (f) {};
	\path[-] 
	(a) edge (b)
	(b) edge (c)
		edge (f)
	(c) edge (d)
	(d) edge (e)
;	 
\end{tikzpicture}
}
\caption{A meta-edge from square $x$ to square $y$}\label{fig:MetaEdge}
\end{figure}
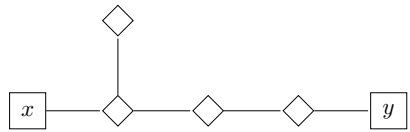

The fact that there is a meta-edge from square $x$ to square $y$ can be described by a formula $\ME(x,y)$ asserting the existence of the above pattern. Notice that such a meta-edge can create a cycle of length $6$, if $x$ and $y$ are connected to the same source, but will not introduce a cycle of length at most $5$.

We call ``diamonds'' the auxiliary vertices that are not sources, and that are connected to a square by a path of length $1$ or $2$. This can be defined by a positive formula $\diam(x)$. 

For a graph to be in $\Gw$, we require the following additional constraints:
\begin{itemize}[align=left]
\item[(partition)] Any vertex is either a source, a square or a diamond, with no overlap.
\item[(cycle)] The only cycles of length at most $5$ are those composed of sources.
\item[(order)]  The meta-edge relation form a strict total order on squares.
\item[(diamonds)] A diamond can be part of at most one meta-edge.
\end{itemize}

As before, we can express all these constraints with positive formulas $\psim$ and $\psip$.
Let us explicit these formulas for rules (partition) and (order). We can use $\neg\Sq(x):=\Ci(x)\vee\diam(x)$ to assert that $x$ is not a square, thanks to rule (partition). This allows us to quantify on squares only, that we will abbreviate $\exists^\Sq$ and $\forall^\Sq$.

$$\arraycolsep=3pt\def\$arraystretch{1.3}
\begin{array}{|c|c|c|}
\hline
\textbf{Constraint} & \bm{\psim} & \bm{\psip} \\
\hline
\text{(partition)} & \forall x.\Ci(x)\vee\Sq(x)\vee\diam(x) & \exists x. \Sq(x)\wedge\diam(x)\\
\hline
\text{(order)} & \forall^\Sq x,y.\ME(x,y)\vee\ME(y,x)& 
\begin{array}{ll}
&\exists^\Sq x. \ME(x,x)\\
\vee & \exists^\Sq x,y.\ME(x,y)\wedge\ME(y,x)\\
\vee & \exists^\Sq x,y,z. \ME(x,y)\wedge\ME(y,z)\wedge\ME(z,x)
\end{array}\\
\hline
\end{array}$$
\smallskip

The formulas for (cycle) and (diamonds) pose no additional difficulty: the one for (cycle) is similar to the directed case, and the one for (diamonds) only has a $\psip$ component, stating the existence of two meta-edges sharing a diamond.
The formula $\psim$ (resp. $\psip$) will again be the conjunction (resp. disjunction) of its components coming from various rules.
\medskip

A graph $G\in\Gw$ will be an encoding of a word $u\in (A\setminus\{\emptyset\})^*$, by interpreting the squares with the meta-edge order as $(\dom(u),<)$, and a predicate $a(x)$ true if the square $x$ is connected to an $a$-source (resp. with $b,c$).

For instance the following graph encodes the word $ab\x c$, where meta-edges are represented by dashed arrows:
\newdimen\R
\R=.8cm
\begin{figure}[H]
\centering
\scalebox{.8}{
\begin{tikzpicture}[shorten >=1pt,node distance=2cm,on grid,auto,
rect/.style={rectangle,draw,minimum size=6mm},
circ/.style={circle,draw,minimum size=4mm},
source/.style={circle,draw,fill=white,inner sep=1mm}]

    \draw (0:\R) node[source] (a1) {}  
            -- (120:\R) node[source] (a2) {}  -- (240:\R) node[source] (a3) {}
        -- cycle (90:\R) node[above=.1cm] {$a$-sources} ;
    \draw[xshift=2.6\R] (0:\R) node[source] (b0) {} \foreach \x in {90,180,...,359} {
            -- (\x:\R) node[source] (b\x) {}
        } -- cycle (90:\R) node[above=.1cm] {$b$-sources} ;
    \draw[xshift=5.1\R] (0:\R) node[source] (c0) {} \foreach \x in {72,144,...,359} {
            -- (\x:\R) node[source] (c\x) {}
        } -- cycle (90:\R) node[above=.1cm] {$c$-sources} ;
    
        %

            	\node[rect, below left=2.5cm and 1.5cm of a1] (s1) {};
	\node[rect, right= of s1] (s2) {};
	\node[rect, right= of s2] (s3) {};
	\node[rect, right= of s3] (s4) {};
	
	\path[->,dashed]
	(s1) edge (s2)
		 edge[bend right] (s3)
		 edge[bend right] (s4)
	(s2) edge (s3)
		edge[bend right] (s4)
	(s3) edge (s4);
	\path[-]
	(a1) edge (s1)
	(a3) edge (s3)
	(b180) edge (s2)
	(b0)	edge (s3)
	(c0) edge (s4);
\end{tikzpicture}
}
\caption{A graph of $\Gw$ encoding $ab\x c$}\label{fig:Gwundir}
\end{figure}
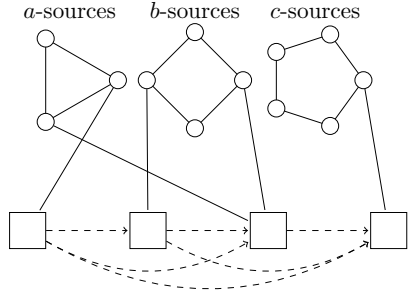

Similarly to the directed case, we will transform the formula $\phiK$ into a formula $\phiKG$ using the following translation:
\begin{itemize}
\item $(x\leq y)^G:=\ME(x,y)\vee x=y$
\item $(x<y)^G:=\ME(x,y)$
\item $(\alpha(x))^G:= \exists x_\alpha.\Ci_\alpha(x_\alpha)\wedge E(x_\alpha,x)$, for $\alpha\in\{a,b,c\}$.
\item $(\exists x.\varphi)^G:=\exists^\Sq x.\varphi^G$
\item $(\forall x.\varphi)^G:=\forall^\Sq x.\varphi^G$
\item $(\varphi\wedge\psi)^G:=\varphi^G\wedge\psi^G$
\item $(\varphi\vee\psi)^G:=\varphi^G\vee\psi^G$
\item $(\neg \varphi)^G:=\neg(\varphi^G)$
\end{itemize}

As before, we define $\phi=\exists \vec x.(\psim\wedge(\phiKG\vee\psip))$, where $\vec x$ is the list of the $12$ source vertices.
The fact that $\phi$ is monotone and $\sem{\phi}\cap \Gw=\GK$, the graphs encoding words in $K$, is proved the same way.

We again show that there is no positive formula for $\phi$ using the $\EF$-game technique according to \Cref{cor:EFgen}. Given words $u$ and $v$ as before, we choose canonical encodings $G_u$ and $G_v$ in $\Gw$, where all diamonds are used in some meta-edge. We choose $G_u$ and $G_v$ of the simplest possible form, for instance without parallel meta-edges. We will show that Duplicator can use his winning strategy from $\EF_n(u,v)$ to win in the $(G_u,G_v)$ arena as well.
As before, sources and squares pose no problem, and the strategy can be directly copied from words to graphs. However, there is a new subtlety to take care of: Spoiler can now play on diamond vertices in $G_u$ or $G_v$. Since by rule (diamonds), any diamond is part of exactly one meta-edge from a square $x$ to a square $y$. In order to answer this move in the graph game, Duplicator will look at what happens in the word game if Spoiler plays positions $x$ and $y$. Duplicator's winning strategy gives answers $x'$ and $y'$ to these moves. He can now answer the corresponding diamond in the other graph, in the meta-edge between squares $x'$ and $y'$, at the same relative position as the diamond originally played by Spoiler. Since playing $2$ moves on words can be necessary to answer one move on graphs, Duplicator will only be able to play $n/2$ rounds in the game on $G_u,G_v$. This is enough to show that for any $n$ there are $G_u\in\sem{\phi}$ and $G_v\notin\sem{\phi}$ such that Duplicator wins $\EF_n(G_u,G_v)$, thereby proving that $\sem{\phi}$ is not definable in $\FOp$.

This concludes the proof scheme of the following theorem:

\begin{thm}
Lyndon's Theorem fails on finite undirected graphs: there is an FO-definable monotone class of undirected graphs that is not definable with a negation-free formula.
\end{thm}

\section{Undecidability of $\FOp$-definability for regular languages}\label{sec:undec}

In this section, we are back to considering only finite words. Our goal will be to prove the following Theorem:

\begin{thm}\label{thm:undec}
The following problem is undecidable: given $L$ a regular language on a powerset alphabet, is $L$ $\FOp$-definable?
\end{thm}

Notice that since FO-definability is decidable for regular languages, this theorem could be equivalently stated with an input language given by an FO formula.

We will a start with an informal proof sketch to convey the main ideas of the proof, before going to the technical details.

\subsection{Proof sketch}\label{sec:sketch}

The proof proceeds by reduction from the Turing Machine Mortality problem, known to be undecidable \cite{Hooper66}. A deterministic Turing Machine (TM) is \emph{mortal} if there is a uniform bound $n\in\N$ on the length of its runs, starting from any arbitrary configuration. 

Given a machine $M$, we want to build a regular language $\LM$ such that $\LM$ is $\FOp$-definable if and only if $M$ is mortal.
\medskip

\noindent\textbf{Configuration words}

The intuitive idea is that $\LM$ will mimic the language $(\cl a\cl b\cl c)^*$ from \Cref{def:K}, but the letters will be replaced by words encoding configurations of $M$. We therefore design an ordered alphabet $A$ and a language $C$ of configurations words such that words from $C$ encode configurations of $M$. These words will be of three possible types $1,2,3$, playing the role of the letters $a,b,c$ of the language $K$. This partitions $C$ into $C_1\cup C_2\cup C_3$. We guarantee that the transitions of $M$ will always change the types in the following way: $1\to 2$, $2\to 3$ or $3\to 1$.

Moreover, we design the order $\leq_A$ of the alphabet $A$ so that given two words $u_1,u_2$ from $C$, there is a word $v$ that is bigger (for the order $\leq_A$) than both $u_1$ and $u_2$ if and only if $u_1$ and $u_2$ are consecutive configurations of $M$. Such a word $v$ will be written $\duo{u_1}{u_2}$ in this proof sketch.
\medskip

\noindent\textbf{Language $\LM$}

Finally, the language $\LM$ will be roughly defined as the upward-closure of $(C_1\# C_2\# C_3\#)^*$, where $\#$ is a separator symbol. The only requirement for configuration words appearing in a word of $\LM$ is on their types. Apart from this, the configuration words can be arbitrary, they do not have to form a run of $M$.

We will then use the $\EF$-game technique to show that $\LM$ is $\FOp$-definable if and only if $M$ is mortal.
\medskip

\noindent\textbf{If $M$ not mortal}

The easier direction is proving that if $M$ is not mortal, then $\LM$ is not $\FOp$-definable.
Indeed, if $M$ is not mortal, we can choose an arbitrarily long run $u=u_1\# u_2\#\dots\# u_N$ of $M$, where $u_1\in C_1$ and $N$ is a multiple of $3$ (implying $u_N\in C_3$). We build the word $v=\duo{u_1}{u_2}\# \duo{u_2}{u_3}\#\dots \#\duo{u_{N-1}}{u_N}$, and we verify that $u\in \LM$ and $v\notin \LM$ (because of the number of $\#$ modulo $3$). Then, using the same technique as in the proof of \Cref{lem:notFOp}, with $C_1,C_2,C_3$ playing the role of $a,b,c$, we show that Duplicator wins the $\EF$ game on $u,v$ with $\log(N)$ rounds. Therefore $\LM$ is not $\FOp$-definable, by \Cref{cor:EF}.
\medskip

\noindent\textbf{If $M$ mortal}

The converse direction is more difficult: we have to show that if there is a bound $n$ on the length of runs of $M$ from any configuration, then $\LM$ is $\FOp$-definable.
We will again use the $\EF$-game and \Cref{cor:EF}: we give an integer $m$ (depending only on $n$) such that for any $u\in \LM$ and $v\notin \LM$, Spoiler wins $\EF_m(u,v)$.

Without loss of generality, consider $u=u_0\#u_1\#\dots \# u_N$ a word of $\LM$, where each $u_i$ is in $C$, and $v$ a word not in $\LM$.
To describe the winning strategy of Spoiler, we will first rule out the problems in ``local behaviours'': if $v$ contains a factor containing at most $2$ symbols $\#$ preventing it from belonging to $\LM$, then Spoiler wins easily in a bounded number of rounds by pointing this local inconsistency in $v$, that cannot be mirrored in $u$.
The only remaining problem is the ``long-term inconsistency'' occurring in the previous $\EF$ games: a long factor with two conflicting possible interpretations, each being forced by one of the endpoints. For instance, when dealing with the language $K$, such long-term inconsistencies were exhibited by words of the form $\x\y\z\x\dots\y$, with the first letter being constrained to $a$ and the last one to $c$.
We have to show that contrarily to what happens with the language $K$, or in the case where $M$ is not mortal, Spoiler can now point out such long-term inconsistencies in a bounded number of rounds.

To do that, let us abstract a configuration word $w\in C$ by its \emph{height} $h(w)$: the length of the run of $M$ starting in the configuration $w$. Our mortality hypothesis can be rewritten as: the height of any configuration word is at most $n$. A word $w\in C$ will be abstracted by a single letter $h(w)\in[0,n]$.
We saw that if a word $w'$ is of the form $\duo{w_1}{w_2}$, then $w_1$ and $w_2$ encode consecutive configurations of $M$, so their heights must be consecutive integers $i+1$ and $i$. We will abstract such a word $w'$ by the letter $\duo{i+1}{i}$.
This allows us to design an abstracted version of the EF-game, called the \emph{integer game}, where letters are integers or pairs of integers, and with special rules designed to reflect the constraints of the original $\EF$ game on $(u,v)$. The integer game makes explicit the core combinatorial argument making use of the mortality hypothesis. We show that Spoiler wins this integer game in $2n$ rounds. We finally conclude by lifting this strategy to the original $\EF$-game.
\medskip

This ends the proof sketch, and we now go to the more detailed proof.
\subsection{The Turing Machine Mortality problem}\label{subsec:TMM}

We will start by describing the problem we will reduce from, called Turing Machine (TM) Mortality.

The TM Mortality problem asks, given a deterministic TM $M$, whether there exists a bound $n\in\N$ such that from any finite configuration (state of the machine, position on the tape, and content of the tape), the machine halts in at most $n$ steps. We say that $M$ is \emph{mortal} if such an $n$ exists.

\begin{thmC}[\cite{Hooper66}]
The TM Mortality problem is undecidable.
\end{thmC}

\begin{rem}\label{rem:compactness}
The standard mortality problem as formulated in \cite{Hooper66} does not ask for a uniform bound on the halting time, and allows for infinite configurations, but it is well-known that the two formulations are equivalent using a compactness argument. Indeed, if for all $n\in\N$, the TM has a run of length at least $n$ from some configuration $C_n$, then we can find a configuration $C$ that is a limit of a subsequence of $(C_n)_{n\in\N}$, so that $M$ has an infinite run from $C$.
\end{rem}

Notice that the initial and final states of $M$ play no role here, so we will omit them in the description of $M$. Indeed, we can assume that $M$ halts whenever there is no transition from the current configuration.

Let $M=(\Gamma,Q,\Delta)$ be a deterministic TM, where $\Gamma$ is the alphabet of $M$, $Q$ its set of states, and $\Delta\subseteq Q\times\Gamma\times Q\times\Gamma\times\Dir$ its (deterministic) transition table.

We will also assume without loss of generality that $Q$ is partitioned into $Q_1,Q_2,Q_3$, and that all possible successors of a state in $Q_1$ (resp. $Q_2, Q_3$) are in $Q_2$ (resp. $Q_3,Q_1$). Remark that if $M$ is not of this shape, it suffices to make three copies $Q_1$, $Q_2$, $Q_3$ of its state space, and have each transition change copy according to the 1-2-3 order given above. This transformation does not change the mortality of $M$.

\noindent We will say that $p$ has \emph{type} $i$ if $p\in Q_i$. The \emph{successor type} of $1$ (resp. $2$, $3$) is $2$ (resp. $3$, $1$).
\medskip

Our goal is now to start from an instance $M$ of TM Mortality, and define a regular language $\LM$ such that $\LM$ is $\FOp$-definable if and only if $M$ is mortal.

\subsection{The base language $\Lbase$}
~\\
\noindent\textbf{The base alphabet}

 We define first a \emph{base alphabet} $\Abase$. Words over this alphabet will be used to encode configurations of the TM $M$.
$$\Abase=\Gamma\cup(\Delta\times\Gamma)\cup (\Gamma\times\Delta)\cup(\Delta\times\Gamma\times\Delta)\cup (Q\times\Gamma)\cup \{\#\}.$$ 

We will note  $a_\delta$ (resp. $a^{\delta'}, a_\delta^{\delta'}$) the letters from $\Delta\times\Gamma$ (resp. $\Gamma\times\Delta, \Delta\times\Gamma\times\Delta$), and $[q.a]$ letters of $Q\times\Gamma$.

The letter $[q.a]$ is used to encode the position of the reading head, $q\in Q$ being the current state of the machine, and $a\in \Gamma$ the letter it is reading.

A letter $a_\delta$ will be used to encode a position of the tape that the reading head just left, via a transition $\delta$ writing an $a$ on this position. A letter $a^{\delta'}$ will be used for a position of the tape containing $a$, and that the reading head is about to enter via a transition $\delta'$. We use $\add$  if both are simultaneously true, i.e. the reading head is coming back to the position it just visited.
Finally, the letter $\#$ is used as separator between different configurations.
%
\medskip

\noindent\textbf{Configuration words}

The encoding of a configuration of $M$ is therefore a word of the form (for example): $$a_1a_2\dots (a_{i-1})^{\delta'}[q.a_i](a_{i+1})_\delta\dots a_n.$$ The letter $(a_{i+1})_\delta$ indicates that the reading head came from the right via a transition $\delta=(\_,\_,q,a_{i+1},\leftarrow)$ (where $\_$ is a placeholder for an unknown element). The letter $(a_{i-1})^{\delta'}$ indicates that it will go in the next step to the left via a transition $\delta'=(q,a_i,\_,\_,\leftarrow)$.

A word $u\in (\Abase)^*$ is a \emph{configuration word} if it encodes a configuration of $M$ with no inconsistency. More formally, $u$ is a configuration word if $u$ contains no $\#$, exactly one letter from $Q\times \Gamma$ (the reading head), and either one $a_\delta$ and one $b^{\delta'}$ located on each side of the head, or just one letter $\add$ adjacent to the head. Moreover, the labels $\delta$ and $\delta'$ both have to be coherent with the current content of the tape. 

\begin{rem}\label{rem:Cencoding}
Because we ask these $\delta$ and $\delta'$ labellings to be present, configuration words only encode TM configurations that have a predecessor and a successor configuration.
\end{rem}

The \emph{type} of a configuration word is simply the type in $\{1,2,3\}$ of the unique state it contains.

Let us call $C\subseteq (\Abase)^*$ the language of configuration words. This language $C$ is partitioned into $C_1,C_2,C_3$ according to the type of the configuration word. It is straightforward to verify that each $C_i$ is an FO-definable language.

We can now define the language $\Lbase$. The basic idea is that we want $\Lbase$ to be $(C_1\#C_2\#C_3\#)^*$, but in order to avoid unnecessary bookkeeping later in the proof, we do not want to care about the endpoints being $C_1$ and $C_3$. Let us also drop the last $\#$ which is useless as a separator, and assume that $C_1$ appears at least once, just for the sake of simplifying the final expression. This gives for $\Lbase$ the expression: 
{\small $$(\varepsilon+C_3\#+C_2\#C_3\#)(C_1\#C_2\#C_3\#)^*(C_1+C_1\#C_2+C_1\#C_2\#C_3).$$}
%
Notice that $\Lbase$ cannot verify that the sequence is an actual run of $M$, since it just controls that the immediate neighbourhood of the reading head is valid, and that the types succeed each other according to the 1-2-3 cycle. The tape can be arbitrarily changed from one configuration word to the next.

\subsection{The alphabet $A$}

We now define another alphabet $\Aamb$ (\emph{amb} for ambiguous), consisting of some unordered pairs of distinct letters from $\Abase$. An unordered pair $\{a,b\}$ is in $\Aamb$ if $a$ can be replaced with $b$ in the encodings of two successive configurations of $M$ of the same length (with surrounding letters changing as well). Thus, let $\Aamb$ be the following set of unordered pairs (we note the ``predecessor'' element first to facilitate the reading):
{\small
\begin{itemize}
    \item $\{a_\delta,a\}$, $a\in\Gamma$, $\delta\in\Delta$
    \item $\{a,a^{\delta'}\}$, $a\in\Gamma$, $\delta'\in\Delta$
    \item $\{a^{\delta'},[q.a]\}$, $\delta'=(\_,\_,q,\_,\_)\in\Delta$
    \item $\{\add,[q.a]\}$, $\delta=(\_,\_,p,a,d)\in\Delta$, $\delta'=(p,\_,q,\_,-d)\in\Delta$ 
    \item $\{[p.a],b_\delta\}$, $\delta=(p,a,\_,b,\_)\in\Delta$ 
    \item $\{[p.a],b_\delta^{\delta'}\}$, $\delta=(p,a,q,b,d)\in\Delta$, $\delta'=(q,\_,\_,\_,-d)\in\Delta$
\end{itemize}
}
\noindent where $\_$ stands for an arbitrary element, and $-d$ is the direction opposite to $d$.

Notice that all letters of $\Aamb$ have a clear ``predecessor'' element: even the possible ambiguity regarding letters $\add$ are resolved thanks to the type constraint on transitions of $M$. For readability, we will use the notation $\duo{a}{b}$ instead of $\{a,b\}$, where the upper letter is the predecessor element. 

We can now define the alphabet $A=\Abase\cup \Aamb$, partially ordered by $a<_A b$ if $a\in \Abase, b\in \Aamb, a\in b$. The order $\leq_A$ is the reflexive closure of $<_A$.
For now we use the  general formalism of ordered alphabet for simplicity. We will later describe in \Cref{rem:power_undec} how the construction is easily modified to fit in the powerset alphabet framework.

\subsection{Superposing configuration words}
We will see that thanks to the definition of the alphabet $\Aamb$, two distinct configurations can be ``superposed'', i.e. can be written simultaneously with letters of $A$ including letters of $\Aamb$, if and only if one follows from the other by a valid transition of $M$.

\begin{lem}\label{lem:merge}
If $u_1,u_2\in C$ encode two successive configurations of the same length, then there exists $v\in A^*$ such that $u_1\leq_A v$ and $u_2\leq_A v$.
\end{lem}

\begin{proof}

It suffices to take the letters in $v$ to be the union of letters in $u_1,u_2$ when these\linebreak[4]letters differ.
For instance if $u_1=aab^{\delta'}[p.a]c_\delta c$ and $u_2=aa^{\delta''}[q.b]d_{\delta'}cc$, then
$v=$\linebreak[4]$a\duo{a}{a^{\delta''}}\duo{b^{\delta'}}{[q.b]}\duo{[p.a]}{d_{\delta'}}\duo{c_\delta}{c}c$.
\end{proof}

\begin{lem}\label{lem:succ}
Let $u_1,u_2\in C$, and assume that there exists $v\in A^*$ satisfying $u_1\leq_A v$ and $u_2\leq_A v$.
Then either $u_1=u_2$, or one is the successor configuration of the other.
\end{lem}

\begin{proof}

Let $[p.a]$ and $[q.b]$ be the reading heads in $u_1$ and $u_2$, in positions $i$ and $j$ respectively.

If $i=j$, let us consider the letter $v[i]$, we know that $[p.a]\leq_A v[i]$ and $[q.b]\leq_A v[i]$. Since no letter of the form $\{[p.a],[q.b]\}$ exists in $\Aamb$, we must have $[p.a]=[q.b]$. Let $\delta'=(p,a,p',a',d)$ be the transition of $M$ from $[p.a]$. Let $\lambda=-1$ if $d=\leftarrow$ and $\lambda=1$ if $d=\rightarrow$. By definition of $C$, we must have letters $b_1,b_2$ such that $u_1[i+\lambda]=b_1^{\delta'}$ and $u_2[i+\lambda]=b_2^{\delta'}$, with possible additional $\delta$ subscripts. As before, there is no letter $\{b_1^{\delta'},b_2^{\delta'}\}$ in $\Aamb$, even with optional additional $\delta$ subscripts, so $u_1[i+\lambda]=u_2[i+\lambda]$. Finally, either both $u_1[i-\lambda]$ and $u_2[i-\lambda]$ are letters from $\Gamma$ (if the $\delta$ subscript is present $u_1[i+\lambda]=u_2[i+\lambda]$), or are letters with a $\delta$ subscript. In both cases, again by definition of $\Aamb$, we must have $u_1[i-\lambda]=u_2[i-\lambda]$. The rest of the words $u_1,u_2$ outside of these three positions $\{i_1,i,i+1\}$ are forced to be letters of $\Gamma$ by definition of $C$, so again they cannot vary between $u_1$ and $u_2$, since $\Aamb$ does not contain letters $\{a,b\}$ with $a,b\in\Gamma$. We can conclude that $u_1=u_2$.
\medskip

Consider now the case where $i\neq j$. Assume without loss of generality that $p$ is of type $1$ (no type plays a particular role). Let $\delta_1$ be the transition from $[p.a]$, of type $1\to 2$. Let us called \emph{enriched letter} a letter of the form $a_\delta$, $a^{\delta'}$, or $a_\delta^{\delta'}$. By definition of $\Aamb$, both $u_2[i]$ and $u_1[j]$ must be enriched letters. By definition of $C$, it means $u_1[j]$ is just next to $u_1[i]$ where the reading head is, say without loss of generality $j=i+1$. 
So we have $v[i]=\{[p.a],u_2[i]\}\in\Aamb$, and as we saw, $u_2[i]$ is either predecessor or successor to $[p.a]$ in this case.  Let us assume for now that $u_2[i]$ is successor to $[p.a]$. This means it has a $\delta_1$ subscript. Since $u_2\in C$, this forces the state $q$ to be of type $2$, the target type of $\delta_1$. Consider now the enriched letter $u_1[j]$, which is such that $\{[q.b],u_1[j]\}\in\Aamb$. Either $u_1[j]$ has a superscript $\delta'$ with target $q$, or a subscript $\delta_2$ with source $p$. This latter case is not possible, as from the fact that $u_1\in C$, it would force $p$ to be of type $3$, the target type of $\delta_2$.
So we have indeed that both $u_1[i]$ and $u_1[j]$ are the predecessors of $u_2[i]$ and $u_2[j]$ in the letters $v[i], v[j]$ of $\Aamb$ respectively. It is straightforward to verify that this implies that $u_1$ is the predecessor configuration of $u_2$: the only discrepancies allowed between $u_1$ and $u_2$ outside of positions $i,j$ are of the form $\{a_\delta,a\}$ and $\{a,a^{\delta'}\}$ and do not influence the underlying letter of $\Gamma$. Similarly, in the other case, where $u_2[i]$ is predecessor to $[p.a]$ in the letter $v[i]$, we obtain that $u_1$ is the successor configuration to $u_2$. 
\end{proof}

\begin{lem}\label{lem:only2}
It is impossible to have three distinct words $u_1,u_2,u_3\in C$ and $v\in A^*$ such that for all $i\in\{1,2,3\}$, $u_i\leq_A v$.
\end{lem}

\begin{proof}
By \Cref{lem:succ}, any pair from $\{u_1,u_2,u_3\}$ must encode two consecutive configurations of $M$.
However, since the reading head must move at each step, from $u_1$ to $u_2$ and from $u_2$ to $u_3$, this means the reading head moves either $0$ or $2$ positions between $u_1$ and $u_3$, which yields a contradiction.
\end{proof}

\subsection{The language $\LM$}

We finally define $\LM$ to be the monotone closure of $\Lbase$ with respect to $\leq_A$, so that $\LM$ can contain letters from $\Aamb$.

By Lemma \ref{lem:clA}, since $\LM$ is the monotone closure of a regular language, it is regular (and monotone).

As a side remark, we can observe the following:
\begin{rem}
$\LM$ is FO-definable.
Since it is not crucial to the following, we only give here a rough intuition on why $\LM$ is FO-definable. The language $K$ from Section \ref{sec:K} can be seen as an abstraction of $\LM$, with $a,b,c$ playing the role of $C_1,C_2,C_3$ respectively. In this light, and since $C_1,C_2,C_3$ are all FO-definable, we can use the fact that the language $K$ is FO-definable as well, by \Cref{lem:KFO}, to obtain an FO-formula for $\LM$. We also need \Cref{lem:only2} to guarantee that the equivalent of the letter $\top$ from $K$ never appears.
\end{rem}

We will now prove in the next sections that $\LM$ is $\FOp$-definable if and only if $M$ is mortal, using \Cref{cor:EF}.


\subsection{$M$ not mortal $\implies$ $\LM$ not $\FOp$-definable}\label{sec:impl}

Let $n\in\N$, we aim to build $(u,v)\in \LM\times \overline{\LM}$ such that $u\preceq_n v$.

There is a configuration from which $M$ has a run of length $N+3$, with $N=2^{n+1}+1$.
Let $u=u_0\#u_1\#\dots\#u_N$ be an encoding of this run where each $u_i\in C$, and where we omitted the first and last configurations of the run, which may not be representable in $C$ by Remark \ref{rem:Cencoding}. Here all the $u_i$'s are of the same length, which is the size of the tape needed for this run.

By Lemma \ref{lem:merge}, for each $i\in[0,N-1]$, there exists $v_i\in \Ast$ such that $u_i\leq_A v_i$ and $u_{i+1}\leq_A v_i$.

We build $v=u_0\#v_1\#\dots \#v_{N-2}\#u_N$. Notice that $v\notin \LM$, because the types of  $u_0$ and $u_N$ forces them to be separated by $N-1 \mod 3$ configurations as in $u$, but in $v$ they are separated by $N-2\mod 3$ configurations.
\smallskip

We describe a strategy for Duplicator witnessing $u\preceq_n v$. It is a simple adaptation from the proof of \Cref{lem:notFOp}, so we will just sketch the idea.

Let us consider that initially, there is a pair of initial (resp. final) tokens at the beginning (resp. end) of $u,v$. That is, we choose to start by Spoiler playing his first two moves at the first and last position of $u$, answered correctly by Duplicator at the first and last position of $v$ (otherwise Spoiler easily wins in one move).
We will consider that the initial tokens are ``blue'', and the final ones are ``yellow''.
In the following, a pair of corresponding tokens in $u,v$ will be blue (resp. yellow) if they are at the same distance to the beginning (resp. end) of the word.

When Spoiler plays a token in $u_i$ (resp. $v_i$), Duplicator will look at the color of the closest token in $u_i$, (resp. $v_i$), and answer with a token of the same color, i.e. by playing in $v_i$ (resp. $u_i$) for blue, and in $v_{i-1}$ (resp. $u_{i+1}$) for yellow. Of course, the same strategy applies to tokens played on $\#$ positions, and an arbitrary choice can be made for positions to both colors.

This strategy preserves the following invariant: after $k$ rounds, the number of $\#$ between the last blue token and the first yellow token on the same word ($u$ or $v$) is at least $2^{n-k}$. This invariant guarantees that Duplicator wins the $n$-round game, since this gap will never be empty.

\subsection{$M$ mortal $\implies$ $\LM$ $\FOp$-definable }

Let $M$ be a mortal $TM$, and $n$ be the length of a maximal run of $M$, starting from any configuration.

We will show that $\LM$ is $\FOp$-definable, by giving a strategy for Spoiler in $EF^+_{f(n)}(u,v)$ for any $(u,v)\in \LM\times \overline{\LM}$, where the number of rounds $f(n)$ depends only on $n$, and not on $u,v$.
\medskip

Let $(u,v)\in \LM\times\overline{\LM}$.
Without loss of generality we can assume that $u\in \Lbase$. This is because there exists $u'\in \Lbase$ with $u'\leq_A u$, and we can consider the pair $(u',v)$ instead of $(u,v)$. Indeed, if Spoiler wins on $(u',v)$, then the same strategy is winning on $(u,v)$, where his winning condition only gets easier.

Thus we can write $u=u_0\#u_1\#\dots \# u_N$, where each $u_i$ is in $C$.
Let us also write $v=v_0\#v_1\#\dots\# v_T$, where each $v_i$ does not contain $\#$. Let us emphasize that no assumption is made on $N$ and $T$, they can be any integers.

We will now describe a strategy for Spoiler in $EF^+(u,v)$, that is winning in a number $f(n)$ of rounds only depending on $n$.

\medskip

\noindent\textbf{Ruling out local inconsistencies}

As explained in the proof scheme of Section \ref{sec:sketch}, we will first show that if $v$ presents \emph{local inconsistencies} (that we define here formally via the notion of forbidden local factor), Spoiler can point them out in a bounded number of moves. 

\begin{defi}
Let us call \emph{local factor} a factor containing at most two symbols $\#$. A local factor is \emph{forbidden} if it is not a factor of any word in $\LM$.
\end{defi}
\begin{lem}\label{lem:localfactor}
If $v$ contains a forbidden local factor, Spoiler can win in a constant number of moves (at most $5$).
\end{lem}
\begin{proof} This can be seen by verifying that the language of words containing no forbidden local factors is $\FOp$-definable, with a formula $\philoc$ using at most $5$ nested quantifiers. We sketch here how such a formula $\philoc$ can be built.

Notice that formulas of $\FOp$ can use the letter $\#$ either positively or negatively, since it is not comparable with any other letter.
If $(p,a)\in Q\times\Gamma$, we define $S(p,a):=\{(\sigma,\tau)\in (\Abase)^2\mid \sigma[p.a]\tau\in C\}$ as the possible neighbourhoods of $[p.a]$ in the base alphabet.
Let us define an $\FOp$-formula $\phitrans(x)$ with free variable $x$, to be a formula that verifies that the maximal $\#$-free factor containing position $x$ is in $\cl{C}$, and that this is witnessed by the reading head at position $x$. This formula will verify that $x$ contains some $[p.a]$, that the immediate neighbourhood of $x$ is compatible with $[p.a]$, via a formula $\bigvee_{(\sigma,\tau)\in S(p,a)} \cl \sigma(x-1)\wedge\cl\tau(x+1)$, and that all other letters (not on positions $\{x-1,x,x+1\}$) between the neighbouring $\#$ symbols are in $\cl\Gamma$.

We now give a description of the formula $\philoc$. The formula will state that for all positions $x<y$ of successive $\#$ symbols (i.e. with no $\#$ between them), there must be positions $i_1<x<i_2<y<i_3$, with only two $\#$ symbols in $[i_1,i_3]$, such that $\phitrans(i_1)\wedge\phitrans(i_2)\wedge \phitrans(i_3)$. Additionally, the types of the states in $i_1,i_2,i_3$, must be respectively either $1$-$2$-$3$, $2$-$3$-$1$, or $3$-$1$-$2$.
\end{proof}

From now on, we will therefore assume that $v$ does not contain forbidden local factors.
\medskip

\noindent\textbf{Finding long-term inconsistencies}

We will see how the only remaining cause for $v$ not belonging to $\LM$ is what we called \emph{long-term inconsistencies} in the proof scheme of Section \ref{sec:sketch}. We will formalize this with the notion of non-coherent maximal ambiguous factor.
\medskip

Let us start with an auxiliary definition.
\begin{defi}
A factor $v_i$ of $v$ is \emph{compatible} with type $j\in\{1,2,3\}$ if there exists $u'\in C_j$ with $u'\leq_A v_i$.
The \emph{set-type} of $v_i$ is $\{j\mid v_i\text{ is compatible with }j\}$.
\end{defi}
By \Cref{lem:only2}, each $v_i$ is compatible with at most $2$ distinct types in $\{1,2,3\}$. If $v_i$ is compatible with $2$ types, then one is the predecessor (resp. successor) of the other in the 1-2-3 cycle order, and we call it the \emph{first type} (resp. \emph{second type}) of $v_i$.
We will consider that $v_0$ (resp. $v_T$) is only compatible with $\type(u_0)$ (resp. $\type(u_N)$). Indeed, if Duplicator matches $v_0$ to a word $u_i$ with $i\neq 0$, Spoiler can win the game in the next round, by choosing a $\#$ position before $u_i$ (and same argument for $v_T$).

\begin{defi}
A factor of the form $v_i\# v_{i+1}\#\dots\#v_{j}$ of $v$ is called \emph{ambiguous} if each $v_k$ for $k=i,\dots, j$ is compatible with two types, and the set-types succeed each other in the cycle order $\{1,2\}\to\{2,3\}\to\{3,1\}$. For instance if the set-type of $v_i$ is $\{3,1\}$, then $v_{i+1}$ must have set-type $\{1,2\}$.
An ambiguous factor is \emph{maximal} if it is not contained in a strictly larger ambiguous factor.
\end{defi}

\begin{defi}
A factor $v_i$ of $v$ is called an \emph{anchor} if either $i=0,i=T$ or if $v_{i-1}\#v_i\#v_{i+1}$ is not ambiguous.
\end{defi}

If $v_i$ is an anchor, we can uniquely define its \emph{anchor type}. It is simply its type if $i=0$ or $T$, and otherwise since $v_{i-1}\#v_i\#v_{i+1}$ is not ambiguous, we define the anchor type of $v_i$ to be the only possible type for $v_i$ that does not create an incoherence with its two neighbours. Notice that such a type exists, since we assumed $v$ does not contain forbidden local factors.
\begin{exa}
Assume $v_5$ has set-type $\{2,3\}$, $v_6$ has set-type $\{3,1\}$, and $v_7$ has set-type $\{2,3\}$. Then $v_6$ is an anchor, and its anchor type is $1$. The type $3$ is indeed impossible for $v_6$, since its successor type $1$ is not in the set-type of $v_7$.
\end{exa}

Notice that if Duplicator maps an anchor $v_i$ to a word $u_j$ such that $\type(u_j)$ is not the anchor type of $v_i$, then Spoiler can win in at most $5$ moves, by pointing to a contradiction with the immediate neighbourhood of $v_i$.

\begin{defi}
Let $v_i\# v_{i+1}\#\dots\#v_{j}$ be a maximal ambiguous factor. It is called \emph{coherent} if $v_{i-1}\# v_{i}\#\dots\#v_{j+1}$ is an infix of $\LM$.
\end{defi}

Such a coherent factor will be witnessed by the anchor types of $v_{i-1}$ and $v_{j+1}$: this means that $v_i\# v_{i+1}\#\dots\#v_{j}$ is coherent if the anchor types at $v_{i-1}$, $v_{j+1}$ are either both concatenable with the first type of both $v_i$, $v_j$, or are both concatenable with their second type. Here by ``concatenable'', we mean to respect the $1$-$2$-$3$ order, for instance type $3$ must be followed by type $1$.

\begin{exa}
Let $w=v_i\# v_{i+1}\#\dots\#v_{j}$ be a maximal ambiguous factor, where $v_i$ has set-type $\{1,2\}$ and $v_j$ has set-type $\{2,3\}$. Notice that this implies $j\equiv i+1\mod 3$. Assume $v_{i-1}$ has anchor type $1$, so it is concatenable with the second type of $v_i$. This means that for $w$ to be coherent, we need $v_{j+1}$ to have anchor type $1$, in order to be concatenable with the second type of $v_j$ as well.
\end{exa}

\begin{lem}\label{lem:notcoh}
$v$ contains a maximal ambiguous factor $w$ that is not coherent.
\end{lem}
\begin{proof}
Assume that all maximal ambiguous factors of $v$ are coherent. Since $v$ does not contain forbidden local factors, we have that the anchor types of two consecutive anchors follow the $1$-$2$-$3$ order. This means that the anchor types, together with the coherence of maximal ambiguous factors, give us a witness that $v\in \LM$. This witness is a word of $\Lbase$, obtained by choosing the anchor types on all anchors, and either the first type or the second type uniformly in maximal ambiguous factors, as fixed by the anchors at the extremities. Since we know that $v\notin \LM$, this is a contradiction.
\end{proof}
We are now ready to describe Spoiler's strategy.
Spoiler starts by placing two tokens delimiting a maximal ambiguous factor $w$ that is not coherent, as obtained in \Cref{lem:notcoh}. Because $w$ is not coherent, Duplicator is forced to answer with the first type for one of these tokens, and with the second type for the other: otherwise Spoiler immediately wins by exposing a local inconsistency with the anchors delimiting $w$.
\medskip

We will now show how Spoiler can win on such a factor $w$, starting with these two tokens. For this, we will introduce the notion of height of a configuration word, and the integer game that will abstract the $\EF$-game on $w$. This is where we finally make use of the hypothesis that $M$ is mortal.
\medskip

\noindent\textbf{Abstracting words by integers}

Recall that from the mortality assumption, we have a bound $n$ on the length of runs of $M$, starting from any configuration. 

\begin{defi}
If $u\in C$ is a configuration word, its \emph{height} $h(u)$ is the length of the run starting in $u$, and not going outside of the tape specified in $u$.
\end{defi}

If $u\in C$, let us also define its $n$-approximation $\alpha_n(u)$ as the maximal word in $(\Abase)^{\leq n}\cdot (Q\times\Gamma)\cdot (\Abase)^{\leq n}$ that is an infix on $u$. That is, we remove letters whose distance to the reading head is bigger than $n$.

Here are a few properties of the height:
\begin{lem}\label{lem:height}
For all $u\in C$, the following hold:
\begin{itemize} 
    \item  $0\leq h(u)<n$.
    \item For all $x,y\in \Gamma^*$, we have $h(xuy)\geq h(u)$.
    \item $h(u)=h(\alpha_n(u))$.
    \item If $v\in C$ is the successor configuration of $u\in C$, then $h(v)=h(u)-1$.
\end{itemize}
\end{lem}

\begin{proof}
The first item is a consequence of the fact that $M$ is mortal with bound $n$. Notice that the inequalities are strict because of Remark \ref{rem:Cencoding}: words from $C$ must have a predecessor and a successor configuration. The second item comes from the fact that the run of length $h(u)$ starting in $u$ is still possible when adding a context $x,y$, which is not affected by this run.
The third item uses the fact that a run can only visit the $n$-approximation of $u$, so the context outside of $\alpha_n(u)$ does not affect the height $h(u)$.
The fourth item is a basic consequence of the definition of the height. 
\end{proof}

\begin{cor}\label{cor:heightFO}
The height of a configuration word $u$ is an $\FOp$-definable property, i.e for all $k\in\N$ there exists an $\FOp$ formula $h_k$ such that $h_k$ accepts a configuration word $u\in C$ if and only if $h(u)=k$.
\end{cor}

\begin{proof}
By \Cref{lem:height}, the formula $h_k$ can simply use a lookup table to verify that $\alpha_n(u)$ is of height $k$, using a finite disjunction listing possibilities for $\alpha_n(u)$ being of height $k$. When evaluated on $A^*$, the formula $h_k$ will accept the monotone closure of configuration words of height $k$.
\end{proof}

\begin{rem}
We use here the fact that computation is done locally around the reading head to obtain \Cref{cor:heightFO}. This seems to make Turing Machines more suited to this reduction than e.g. cellular automata, where computation is done in parallel on the whole tape.
\end{rem}

Thanks to the height abstraction, we will show that we can focus on playing a special kind of abstracted EF-game. 
\medskip

\noindent\textbf{The integer game}

The idea is to abstract a configuration word $u_i\in C$ by its height $h(u_i)$. If $u'$ is the predecessor configuration of $u''$, and $u',u''\leq_A v$, we will abstract the word $v$ by $\duo{h(u')}{h(u'')}$.

Let $\Sigmabase=[0,n]$ and $\Sigmaamb=\{\duo{i}{i-1}\mid 1\leq i\leq n\}$. Let $\Sigma=\Sigmabase\cup\Sigmaamb$, ordered by $i\leq_\Sigma\duo{i}{i-1}$ and $i-1\leq_\Sigma\duo{i}{i-1}$ for all $\duo{i}{i-1}\in\Sigmaamb$.

We define the $n$-integer game as follows:
It is played on an arena $(U,V)$ with $U\in (\Sigmabase)^*$ and $V\in(\Sigmaamb)^*$.
If we note $i$ (resp. $j$) the first (resp. last) letter of $U$, then the first (resp. last) letter of $V$ is $\duo{i}{i-1}$ (resp. $\duo{j+1}{j}$).

The rest of the rules is very close to those of $EF^+(U,V)$: in each round, Spoiler plays a token in $U$ or $V$, Duplicator has to answer with a token in the other word, while maintaining the order between tokens, and the constraint that the label of a token in $U$ is $\leq_\Sigma$-smaller than the label of its counterpart in $V$. We add an additional \emph{neighbouring constraint} for Duplicator: consecutive tokens in one word must be related to consecutive tokens in the other, and in this case, if two tokens of $V$ are in consecutive positions labelled $\duo{i}{i-1}\duo{j}{j-1}$, the corresponding tokens in $U$ must be either labelled $i,j$ or $i-1,j-1$. A mix $i,j-1$ or $i-1,j$ is not allowed.

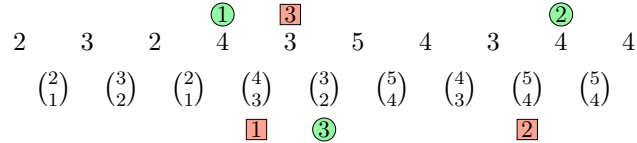
\begin{figure}[H]
\begin{center}
\scalebox{.9}{
\begin{tikzpicture}[eve/.style={circle,draw=black, fill=evegreen!60, inner sep=0,minimum size=.34cm}, adam/.style={draw,fill=adamred!60,inner sep=0, minimum width=.3cm,minimum height=.3cm]}]

\def\s{1}  
\foreach \k/\i in {0/1,1/2,2/1,3/3,4/2,5/4,6/3,7/2,8/3,9/3}
{
\node at (\k*\s,0) {$\pgfmathadd{\i}{1}\pgfmathprintnumber{\pgfmathresult}$};
}
\foreach \k/\i in {0/1,1/2,2/1,3/3,4/2,5/4,6/3,7/4,8/4}
{

\node at (\k*\s+.5*\s,-.7) {{\large $\duo{\pgfmathadd{\i}{1}\pgfmathprintnumber{\pgfmathresult}}{\i}$}};
}
;

\node[adam] at (4*\s,.4) {\small $3$};
\node[eve] at (4*\s+.5*\s,-1.3) {\small $3$};

\node[adam] at (5*\s+2.5*\s,-1.3) {\small $2$};
\node[eve] at (5*\s+3*\s,.4) {\small $2$};

\node[adam] at (3*\s+.5*\s,-1.3) {\small $1$};
\node[eve] at (3*\s,.4) {\small $1$};

\end{tikzpicture}
}
\caption{A position of the integer game.}\label{fig:intgame}
\end{center}
\end{figure}

\begin{lem}\label{lem:ngame}
For all $n\in\N$, Spoiler can win any $n$-integer game in $2n$ rounds.
\end{lem}

\begin{proof}
Let $(u,v)$ be an arena for an $n$-integer game.
We proceed by induction on $n$. 

For $n=1$, the constraints on the game forces $u\in 1(0+1)^*0$ and $v\in\duo{1}{0}^*$.

We can have Spoiler play on the last occurrence of $1$ in $u$, and on the successor position labelled $0$. Duplicator cannot respond to these two moves while respecting the neighbouring constraint, so Spoiler wins in $2$ moves.

Assume now that for some $n\geq 1$, Spoiler wins any $n$-integer game in $2n$ moves, and consider an $(n+1)$-integer game arena $(u,v)$.
If the letters $n+1$ and $\duo{n+1}{n}$ do not appear in $u,v$ respectively, then Spoiler can win in $2n$ moves by induction hypothesis.

If the letter $n+1$ does not appear in $u$, then let $y$ be the first position labelled $\duo{n+1}{n}$ in $v$. By definition of the integer game $y$ cannot be the first position of $v$, otherwise $u$ should start with $n+1$. We will choose position $y$ in $v$ for the first move of Spoiler, let $x$ be the position in $u$ answered by Duplicator, we have $u[x]=n$. We can assume that $x$ is not the first position of $u$, otherwise Spoiler can win in the next move.
If Spoiler were to play $x-1$ in $u$, with $u[x-1]=i$, by the neighbouring constraint Duplicator would be forced to answer $y-1$ in $v$, with label $\duo{i+1}{i}$. This shows that the words $u[..x-1]$ and $v[..y-1]$ form a correct $n$-integer arena, as the integer $n+1$ is not present anymore, and all other constraints are respected. Therefore, Spoiler can win by playing $2n$ moves in these prefixes. This gives a total of $2n+1$ moves in the original $(n+1)$-integer game.

Finally, if the letter $n+1$ does appear in $u$, Spoiler starts by playing the position $x$ in $u$ corresponding to the last occurrence of $n+1$ in $u$. Duplicator must answer a position $y$ labelled $\duo{n+1}{n}$. Notice that neither $x$ nor $y$ can be a last position, so $u[x+1]$ and $v[y+1]$ are well-defined. As before, using the neighbouring constraint, we know that if $i=u[x+1]$, then $v[y+1]=\duo{i}{i-1}$.
Therefore, the words $u[x+1..]$ and $v[y+1..]$ form an $(n+1)$-integer game arena, and moreover the letter $n+1$ does not appear in $u[x+1..]$ (by choice of $x$). Using the precedent case, we know that Spoiler can win from there in $2n+1$ moves, playing only on $u[x+1..]$ and $v[y+1..]$. This gives a total of $2n+2$ moves in the original $(n+1)$-integer game, thereby completing the induction proof. 
\end{proof}

\begin{rem}\label{rem:revgame} Lemma \ref{lem:ngame} still holds if the definition of $n$-integer game is generalized to include the symmetric case where, if we note $i,j$ the first and last letters of $U$ respectively, $V$ starts with $\duo{i+1}{i}$ and ends with $\duo{j}{j-1}$. Indeed, it suffices to consider the mirrored images of $U$ and $V$ to show that Spoiler wins in the same amount of rounds.
\end{rem}

We can now see how to use \Cref{lem:ngame} (together with \Cref{rem:revgame}) to conclude the proof.

\medskip

\noindent\textbf{Lifting the strategy to the original $\EF$-game}

Let us go back to the $\EF$-game on $u,v$, where Spoiler has placed two tokens delimiting a non-coherent maximal ambiguous factor $w$ in $v$.

Spoiler can now play only between these existing tokens, and import the strategy from the integer game, by abstracting each word $u_i$ by its height and each word $v_j$ by $\duo{h(u')}{h(u'')}$, where $u',u''\in C$ are such that $u'\leq_A v_j$, $u''\leq_A v_j$, and  $h(u')=1+h(u'')$. Each factor of $u_i,v_j$ delimited by $\#$ corresponds to a single position in the abstracted integer game. Spoiler can for instance mimic a move of the integer game by playing in the first position of the corresponding factor in $u_i$ or $v_j$, i.e. just after a $\#$-labelled position.

\begin{lem}
If Duplicator does not comply with the rules of the integer game, then Spoiler can win in at most $\log n$ rounds.
\end{lem}
\begin{proof}
Assume $u_i$ is matched to $v_j$, but there is no $u'\leq_A v_j$ such that $h(u_i)=h(u')$. It means that $\alpha_n(u)$ is not compatible with $v$, and this can be exploited by Spoiler using $\log n$ rounds (with a dichotomy strategy, or $n$ rounds with a naive strategy). Thus by Lemma \ref{lem:height}, Spoiler can enforce the basic rule of the integer game, stating that if integer $t$ is matched to $\duo{s+1}{s}$, then $t=s+1$ or $t=s$. Using the correspondence between $\EF$-games and $\FOp$-definability, this property can also be seen via \Cref{cor:heightFO}.

If neighbours are matched with non-neighbours, then it suffices for Spoiler to point the two $\#$ positions between the non-neighbours, that cannot be matched in the other word, so he wins in $2$ moves.
We show that the rest of the neighbourhood rule is also enforced. Assume $u_i\#u_{i+1}$ is matched to $v_j\#v_{j+1}$.
Assume $\type(u_i)$ is the first (resp. second) type of $v_j$ while $\type(u_{i+1})$ is the second (resp. first) type of $v_{j+1}$. By definition of $\LM$, $\type(u_{i+1})$ must be the successor type of $\type(u_i)$, for instance without loss of generality, $\type(u_i)=1$ and $\type(u_{i+1})=2$.
Then, the set-type of $v_i$ is $\{1,2\}$ (resp. $\{3,1\}$) and the set-type of $v_{j+1}$ is $\{1,2\}$ (resp. $\{2,3\}$). This contradicts the fact that $v_i\#v_{i+1}$ is part of an ambiguous factor, as set-types should follow each other in the order $\{1,2\}$-$\{2,3\}$-$\{3,1\}$.
\end{proof}

Combining these arguments and by \Cref{lem:ngame} and \Cref{rem:revgame}, we obtain that following this strategy, Spoiler will win in at most $f(n)=2+2n+\log n+5$ rounds, by either winning the related $n$-integer game, or exploiting local inconsistencies if Duplicator fails to comply to the rules of the $n$-integer games. The $+5$ comes from \Cref{lem:localfactor} about forbidden local factors.

Using \Cref{cor:EF}, we obtain that $\LM$ is $\FOp$-definable, with a formula of quantifier rank at most $f(n)$.

\begin{rem}\label{rem:power_undec}
The alphabet $A$ can be turned into a powerset alphabet, by adding all subsets of $\Abase$ absent from $\Aamb$, rejecting any word containing $\emptyset$ but no new non-empty subset, and accepting any word containing a new non-empty subset. This shows that this undecidability result still holds in the special case of powerset alphabets.
\end{rem}

This concludes the proof of \Cref{thm:undec}, up to the proof of \Cref{lem:ngame} which is done in the next section.

\subsection*{Conclusion}
We believe this paper gives an example of fruitful interaction between automata theory and model theory. Indeed, a classical result of model theory, the failure of Lyndon's theorem on finite structures, has been greatly simplified by using the toolbox of regular languages. Moreover, our investigation of this question via regular languages also brings a new result: the failure of Lyndon's theorem on finite graphs. Conversely, this question coming from model theory, when considered on regular languages, yields the first (to our knowledge) natural fragment of regular languages with undecidable membership problem, and opens new techniques for proving undecidability of expressibility in positive logics. We hope that the tools developed in this paper can be further used in both fields, and that this will encourage more interactions of this form in the future.

Some questions remain open: for instance what happens for partially ordered alphabets that do not embed the one used in our counter-example language? Can a simpler counter-example be found, or is this structure with two layers of three elements each necessary?

We made the signatures needed for our constructions explicit, in several formulations of Lyndon's theorem, and compared them to the ones used in \cite{AjtaiGurevich87,Stol95,RosenPHD}. It can be investigated whether further improvements on these signatures are possible.

We can also investigate the expressive power of fragments of $\FOp$ such as the two-variable fragment, or fragments defined by bounding the quantifier rank.

In the short term, we are interested in extending these techniques to the framework of cost functions, see \cite{Kup14,Erratum}, and to other extensions of regular languages.
\bigskip

\noindent\textbf{Acknowledgements.} I am grateful to Thomas Colcombet for bringing this topic to my attention, and in particular for asking the question $\FOp$ $\stackrel{?}{=}$ monotone FO, as well as for many interesting exchanges. Thanks also to Amina Doumane and Sam Van Gool for helpful discussions, and to the anonymous reviewers, as well as Anupam Das, Natacha Portier, and Bruno Pasqualotto Cavalar for their comments on earlier versions of this document.

\bibliography{biblio}

\begin{thebibliography}{KVB12}

\bibitem[AG87]{AjtaiGurevich87}
Miklos Ajtai and Yuri Gurevich.
\newblock Monotone versus positive.
\newblock {\em J. ACM}, 34(4):1004–1015, October 1987.

\bibitem[AG97]{AlGur97}
Natasha Alechina and Yuri Gurevich.
\newblock {\em Syntax vs. Semantics on Finite Structures}, pages 14--33.
\newblock Springer Berlin Heidelberg, Berlin, Heidelberg, 1997.

\bibitem[Boj04]{MSOU}
Miko{\l}aj Boja{\'{n}}czyk.
\newblock A bounding quantifier.
\newblock In Jerzy Marcinkowski and Andrzej Tarlecki, editors, {\em Computer
  Science Logic}, pages 41--55, Berlin, Heidelberg, 2004. Springer Berlin
  Heidelberg.

\bibitem[Col11]{Green}
Thomas Colcombet.
\newblock Green's relations and their use in automata theory.
\newblock In {\em Language and Automata Theory and Applications - 5th
  International Conference, {LATA} 2011, Tarragona, Spain, May 26-31, 2011.
  Proceedings}, volume 6638 of {\em Lecture Notes in Computer Science}, pages
  1--21. Springer, 2011.

\bibitem[Col12]{CostFun}
Thomas Colcombet.
\newblock Regular cost functions, part i: Logic and algebra over words.
\newblock volume~9, 12 2012.

\bibitem[Col13]{Magnitude}
Thomas Colcombet.
\newblock Magnitude monadic logic over words and the use of relative internal
  set theory.
\newblock In {\em 2013 28th Annual ACM/IEEE Symposium on Logic in Computer
  Science}, pages 123--123, 2013.

\bibitem[DG08]{DG08}
Volker Diekert and Paul Gastin.
\newblock First-order definable languages.
\newblock In {\em Logic and Automata: History and Perspectives, Texts in Logic
  and Games}, pages 261--306. Amsterdam University Press, 2008.

\bibitem[FSS81]{Circuits84}
M.~{Furst}, J.~B. {Saxe}, and M.~{Sipser}.
\newblock Parity, circuits, and the polynomial-time hierarchy.
\newblock In {\em 22nd Annual Symposium on Foundations of Computer Science
  (sfcs 1981)}, pages 260--270, 1981.

\bibitem[GS92]{GrigniSipser}
Michelangelo Grigni and Michael Sipser.
\newblock Monotone complexity.
\newblock In {\em Poceedings of the London Mathematical Society Symposium on
  Boolean Function Complexity}, page 57–75, USA, 1992. Cambridge University
  Press.

\bibitem[Hoo66]{Hooper66}
Philip~K. Hooper.
\newblock The undecidability of the turing machine immortality problem.
\newblock {\em Journal of Symbolic Logic}, 31(2):219–234, 1966.

\bibitem[Kam68]{Kamp}
Hans~W. Kamp.
\newblock {\em Tense Logic and the Theory of Linear Order}.
\newblock Phd thesis, University of Warsaw, 1968.

\bibitem[Kup]{Erratum}
Denis Kuperberg.
\newblock Erratum for \cite{Kup14}.
\newblock URL: \url{perso.ens-lyon.fr/denis.kuperberg/papers/Erratum.pdf}.

\bibitem[Kup14]{Kup14}
Denis Kuperberg.
\newblock Linear temporal logic for regular cost functions.
\newblock {\em Logical Methods in Computer Science}, 10(1), 2014.

\bibitem[Kup21]{Kup21}
Denis Kuperberg.
\newblock Positive first-order logic on words.
\newblock In {\em 36th Annual {ACM/IEEE} Symposium on Logic in Computer
  Science, {LICS} 2021, Rome, Italy, June 29 - July 2, 2021}, pages 1--13.
  {IEEE}, 2021.

\bibitem[KVB12]{KV12}
Denis Kuperberg and Michael Vanden~Boom.
\newblock On the expressive power of cost logics over infinite words.
\newblock In {\em Automata, Languages, and Programming}, pages 287--298,
  Berlin, Heidelberg, 2012. Springer Berlin Heidelberg.

\bibitem[Lib04]{Libkin04}
Leonid Libkin.
\newblock {\em Elements of Finite Model Theory}.
\newblock Springer, August 2004.

\bibitem[LSS96]{PosP}
C.~Lautemann, T.~Schwentick, and I.~A. Stewart.
\newblock On positive p.
\newblock In {\em 2012 IEEE 27th Conference on Computational Complexity}, page
  162. IEEE Computer Society, 1996.

\bibitem[Lyn59]{Lyndon59}
Roger~C. Lyndon.
\newblock Properties preserved under homomorphism.
\newblock {\em Pacific J. Math.}, 9(1):143--154, 1959.
\newblock URL: \url{https://projecteuclid.org:443/euclid.pjm/1103039459}.

\bibitem[MP71]{McNaughtonPapert}
Robert McNaughton and Seymour~A. Papert.
\newblock {\em Counter-Free Automata (M.I.T. Research Monograph No. 65)}.
\newblock The MIT Press, 1971.

\bibitem[PST89]{genstar}
Jean-Eric Pin, Howard Straubing, and Denis Therien.
\newblock New results on the generalized star-height problem.
\newblock volume 349, pages 458--467, 02 1989.

\bibitem[PZ19]{PlaceZeitoun}
Thomas Place and Marc Zeitoun.
\newblock Going higher in first-order quantifier alternation hierarchies on
  words.
\newblock {\em J. ACM}, 66(2), March 2019.

\bibitem[Ros95]{RosenPHD}
{\em Finite model theory and finite variable logic}.
\newblock Phd thesis, University of Pennsylvania, 1995.

\bibitem[Ros08]{Rossman08}
Benjamin Rossman.
\newblock Homomorphism preservation theorems.
\newblock {\em J. ACM}, 55, 07 2008.

\bibitem[Sch65]{Schutz}
M.P. Schützenberger.
\newblock On finite monoids having only trivial subgroups.
\newblock {\em Information and Control}, 8(2):190 -- 194, 1965.
\newblock URL:
  \url{http://www.sciencedirect.com/science/article/pii/S0019995865901087}.

\bibitem[Ste94]{PosNP}
Iain~A. Stewart.
\newblock {Logical Description of Monotone NP Problems}.
\newblock {\em Journal of Logic and Computation}, 4(4):337--357, 1994.

\bibitem[Sto95]{Stol95}
Alexei~P. Stolboushkin.
\newblock Finitely monotone properties.
\newblock In {\em LICS, San Diego, California, USA, June 26-29, 1995}, pages
  324--330. {IEEE} Computer Society, 1995.

\end{thebibliography}
\bibliographystyle{alphaurl}
\end{document}